\documentclass[aps, pra,twocolumn, showpacs,superscriptaddress,10pt,floatfix]{revtex4-2}

\usepackage[utf8]{inputenc}
\usepackage[T1]{fontenc}
\usepackage{verbatim}
\usepackage{bm}
\usepackage{graphicx}
\usepackage{grffile}
\usepackage{amsthm}
\usepackage{array}
\usepackage{epsfig}
\usepackage{amsmath,amssymb,amsfonts,mathtools,bbm,MnSymbol,mathrsfs,xfrac}
\usepackage{longtable}
\usepackage{tabularx}
\usepackage{colortbl}
\usepackage{hhline}
\usepackage[table]{xcolor}
\usepackage[breaklinks=true,colorlinks=true,linkcolor=blue,urlcolor=blue,citecolor=blue]{hyperref}
\usepackage{comment}
\usepackage{color}
\usepackage[makeroom]{cancel}
\usepackage{tikz}
\usepackage{graphicx}
\DeclareGraphicsExtensions{.pdf,.png,.jpg}
\usepackage{algorithm2e}

\usepackage{soul}
\pdfoutput=1

\renewcommand{\[}{\begin{equation}}
\renewcommand{\]}{\end{equation}}

\newcommand{\ket}[1]{|#1\rangle}
\newcommand{\bra}[1]{\langle#1|}

\newcommand{\pro}[2]{|#1\rangle\langle#2|}
\newcommand{\mean}[1]{\langle#1\rangle}

\newcommand{\tr}{\mathrm{tr}}

\newcommand{\at}[2][]{#1|_{#2}}
\newcommand{\R}{{\hat{\rho}}}

\newcommand{\an}{\hat{a}}
\newcommand{\h}{\hat{H}}

\newcommand{\U}{\hat{U}}

\newcommand{\floor}[1]{\left\lfloor #1 \right\rfloor}

\newtheorem{lemma}{Lemma}
\newtheorem{theorem}{Theorem}
\definecolor{dfcol}{cmyk}{1, 0.2108, 0.13, 0.3}
\newcommand{\df}[1]{\ifthenelse{\boolean{}}{\textcolor{dfcol}{[{\bf DF}: #1]}}{}}

\begin{document}
\title{Fundamental Scaling Limit in Critical Quantum Metrology}
\author{Ju-Yeon Gyhm}
\email{kjy3665@snu.ac.kr}
\affiliation{Department of Physics and Astronomy, Seoul National University, 1 Gwanak-ro, Seoul 08826, Korea}

\author{Hyukjoon Kwon}
\email{hjkwon@kias.re.kr}
\affiliation{School of Computational Sciences, Korea Institute for Advanced Study, Seoul 02455, South Korea}

\author{Myung-Joong Hwang}
\email{myungjoong.hwang@duke.edu}
\affiliation{Division of Natural and Applied Sciences, Duke Kunshan University, Kunshan, Jiangsu 215300, China}

\begin{abstract}
Critical quantum metrology aims to harness critical properties near quantum phase transitions to enhance parameter estimation precision. However, critical slowing down inherently limits the achievable precision within a finite evolution time. To address this challenge, we establish a fundamental scaling limit of critical quantum metrology with respect to the total evolution time. We find that the winding number of the system's phase space trajectory determines the scaling bound of quantum Fisher information. Furthermore, we demonstrate that the exponential scaling of the quantum Fisher information can be obtained, and for this, it is necessary to increase the winding number by the total evolution time. We explicitly construct a time-dependent control to achieve optimal scaling from a simple on-off scheme depending on the system's phase and discuss its topological nature. We highlight that such an exponential scaling of quantum Fisher information remains valid even without reaching the critical point and in the presence of thermal dissipation, albeit with a decreased exponent.
\end{abstract}
\maketitle

Quantum metrology~\cite{Giovannetti.2006, Giovannetti.2011} utilizes quantum mechanical principles to achieve sensing tasks beyond the classical limit~\cite{Chu2023}, which can be applied to various tasks, such as detecting gravitational wave~\cite{Schnabel10}, biological imaging~\cite{Aslam23}, improving radars~\cite{Maccone20}, and exploring the foundations of physics~\cite{Couteau23}. A crucial step in quantum metrology is to encode the parameter to be estimated in a quantum state from which non-classical advantages originate~\cite{Pezze.2018, Kwon.2019, Ge20, Maccone2020squeezingmetrology}. More precisely, the quantum Fisher information (QFI)~\cite{Braunstein.1994} characterizes the performance of quantum metrology, providing the ultimate precision of parameter estimation, as given by the Cramer-Rao bound.

In conventional quantum metrology, the preparation of a quantum probe state with a large QFI has been widely studied~\cite{Mitchell04, Leibfried04, Riedel10, Gross10, Monz11, Pezze.2018, Abhinav19, Maccone2020squeezingmetrology, MunozArias23}. An alternative direction utilizing critical phenomena in quantum systems, known as critical quantum metrology, has recently been investigated~\cite{Zanardi.2008, Macieszczak.2016, Fernandez-Lorenzo.2017, Rams.2018, Roscilde.2018, DiFresco2024, Garbe.2019, Chu.2021, Ilias.2022, Garbe_2022, Gietka.2022, Abah2022, Garbe2022, Hotter.2024} with  experimental demonstrations~\cite{Liu.2021,Ding.2022,Beaulieu.2025}. The key idea is to harness the diverging susceptibility and correlations, such as squeezing and entanglement near the critical point, to enhance metrological precision. In particular, phase transitions occurring in fully connected systems, such as the infinite-range Ising model~\cite{Ribeiro.2007} and the Dicke model~\cite{Emary.2003}, as well as the finite-component system phase transitions occurring in qubit-oscillator systems~\cite{Hwang.2015} and the Kerr resonators~\cite{Felicetti.2020,Beaulieu.2025}, have been playing an important role in developing ideas for critical quantum metrology~\cite{Garbe.2019, Ilias.2022,Chu.2021, Garbe_2022,Gietka.2022,Abah2022, Garbe2022,Hotter.2024}. This is partly because a quadratic bosonic Hamiltonian becomes an exact and effective description of the system in the thermodynamic limit, while the finite-component system realizations enable critical sensing using a small-scale quantum system~\cite{Cai.2021,Chen.2021, Zheng.2023,Wu.2024,Ilias.2024,Beaulieu.2025}. 

\begin{figure}[t]
\begin{center}
\includegraphics[width=.95\linewidth]{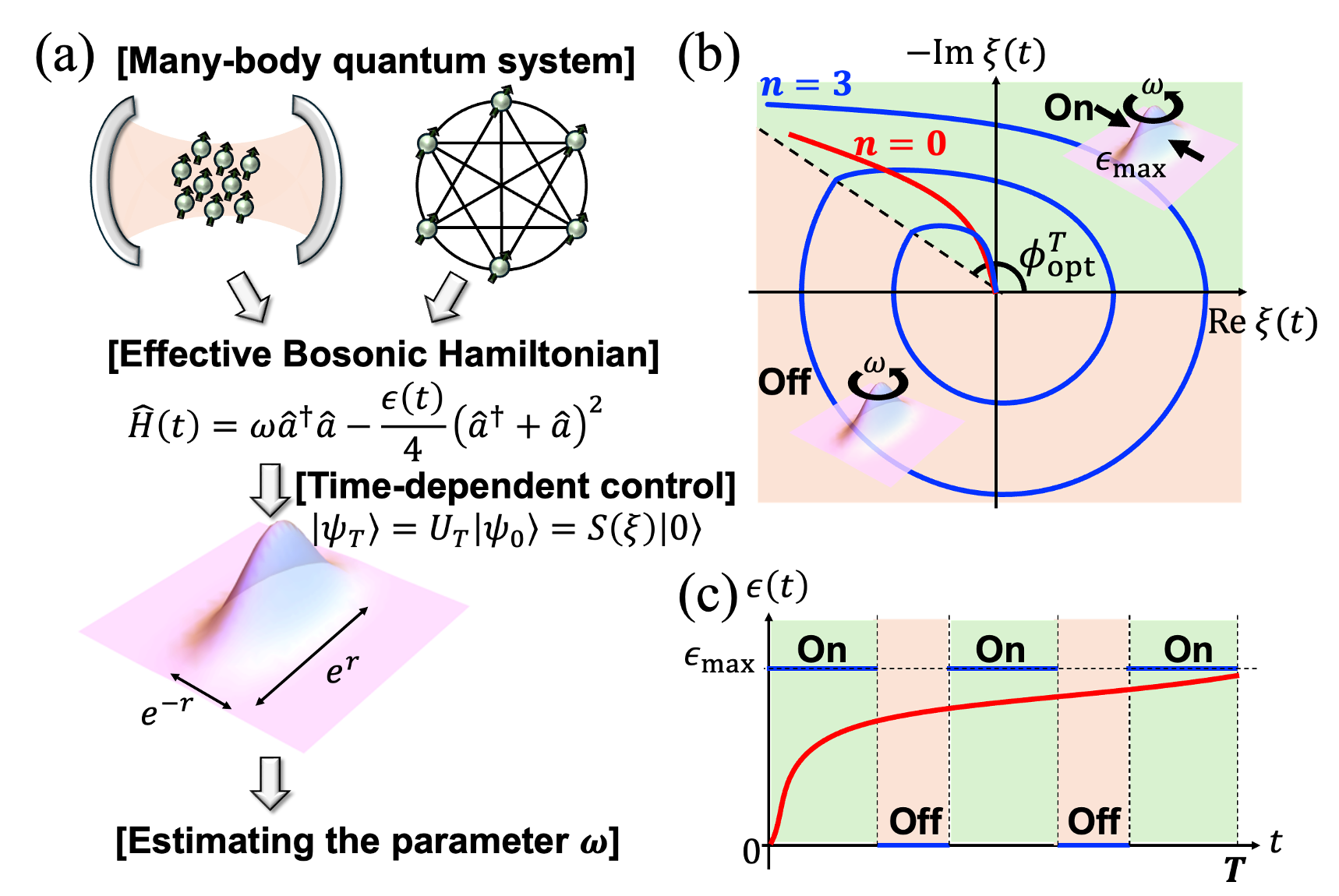}
\caption{(a) Schematics of critical quantum metrology and (b) phase space trajectories with (c) monotonic (red) and on-off (blue) controls. The winding number $n$ of the phase space trajectory governs the scaling limit of QFI. Consequently, increasing the winding number becomes necessary to achieve optimal scaling for a given evolution time $T$, which can be saturated using the on-off scheme.}
\label{Fig:Concept}
\end{center}
\end{figure}

In fully connected models, the Heisenberg scaling of the QFI can be achieved by adiabatically driving the system's ground state toward the critical point~\cite{Garbe.2019}, where the susceptibility diverges. However, this approach suffers from a critical slowing down, as the time required to reach the ground state also diverges, limiting the QFI scaling to $T^4$ in terms of the total evolution time $T$. To address this limitation, recent studies have explored various non-adiabatic protocols such as sudden quench~\cite{Chu.2021, Garbe_2022,Gietka.2022} and periodic driving~\cite{Abah2022, Garbe2022}, revealing that $T^6$ scaling and even exponential scaling can be achieved, respectively. Nevertheless, the ultimate scaling limit of critical quantum metrology with respect to the total evolution time $T$, a key performance metric in the presence of critical slowing down, remains unclear.

In this Letter, we establish a fundamental scaling limit in critical quantum metrology and construct an explicit protocol that achieves this scaling. To this end, we introduce the winding number $n$ determined by the system's phase space trajectory, which plays a central role in both bounding the QFI scaling and establishing optimal control. Our QFI scaling bound applies to both adiabatic~\cite{Garbe.2019} and non-adiabatic protocols~\cite{Chu.2021, Garbe_2022}, as well as non-monotonic protocols~\cite{ Abah2022, Garbe2022}, providing a unified perspective on various scaling behaviors of QFI from different approaches in critical quantum metrology. More surprisingly, we demonstrate that the exponential scaling of the QFI is achievable by determining the optimal winding number for a given time $T$.

We also provide an explicit on-off protocol for achieving optimal scaling, in which the winding number determines how many times to turn the control on and off, depending on the system's phase. The optimal control undergoes a radical transition at the point where the winding number changes, implying its topological nature.

Finally, we highlight that the QFI scaling with total evolution time is determined solely by the winding number rather than the ability to drive the system close to the critical point. As a result, the exponential scaling of the QFI is achievable even without reaching the critical point and remains robust under thermal dissipation.

\emph{Dynamical framework of quantum metrology.---} Suppose a time-dependent Hamiltonian of a quantum system,
\begin{equation}
\h(t) = \omega \h_0 + \h_c(t),
\end{equation}
with $\omega = \omega_0 + \delta \omega$. Here, an unknown parameter $\delta \omega$ is encoded in $\h_0$ with a reference parameter $\omega_0$, and $\h_c(t)$ can be controlled by some tunable parameters. When an initially prepared state $\ket{\psi_0}$ evolves under $\h(t)$ for total evolution time $T$, the unknown parameter $\delta \omega$ can be estimated from the measurement on the final state $\ket{\psi_T}$. In this case, the ultimate precision bound, known as the quantum Cram\'er-Rao bound~\cite{Braunstein.1994}, of estimating $\delta \omega$ is given in terms of the QFI,
\begin{equation}\label{eq:QFI_def}
{\cal F}_\omega = \left. 4 \left[ \langle \partial_\omega \psi_T | \partial_\omega \psi_T \rangle - |\langle \partial_\omega \psi_T | \psi_T \rangle|^2 \right] \right\vert_{\omega = \omega_0}.
\end{equation}

To optimize the performance of quantum metrology, the QFI should be maximized for a given total evolution time $T$, subject to constraints in $\h_c(t)$~\cite{Yuan15, Pang17, Yang22}. This framework incorporates an arbitrary time-dependent protocol $\h_c(t)$ for critical quantum metrology, incorporating conventional quantum metrology as a special case with $\h_c(t) = 0$.

\emph{Critical quantum metrology with a fully connected system.---} We focus on a fully connected system, described by the following effective Hamiltonian~\cite{Hwang.2015,Bakemeier.2012} ($\hbar = 1$),
\begin{equation}
\label{eq:effective_H}
\h(t) = \omega \an^\dagger\an - \frac{\epsilon(t)}{4} (\an^\dagger+\an)^2,
\end{equation}
where $\an$ and $\an^\dagger$ are the bosonic creation and annihilation operators, respectively, and a single control parameter $\epsilon(t)$ induces squeezing. The critical point lies at $\epsilon = \omega$, where the energy spectrum becomes continuous with the closing energy gap. We limit our analysis to the symmetric phase, $0 \leq \epsilon(t)\leq \omega$, where the order parameter remains zero prior to the onset of spontaneous symmetry breaking~\cite{Hwang.2015}.

We assume that the system is initially prepared in the vacuum state $\ket{\psi_0} = \ket{0}$ with $\epsilon = 0$. As the Hamiltonian is quadratic in $\an$ and $\an^\dagger$, the final state becomes a Gaussian state~\cite{RevModPhys.84.621}, represented as
\begin{equation}
\ket{\psi_t} = \hat{S}(\xi(t)) \ket{0} = e^{\frac{1}{2}(\xi^*(t) \hat{a}^2 - \xi(t) \hat{a}^{\dagger 2})} \ket{0},
\end{equation}
where $\hat{S}(\xi(t))$ is the squeeze operator with $\xi(t) = r(t) e^{-i\varphi(t)}$. From the Schr{\"o}dinger equation, $ \frac{d \ket{\psi_t}}{d t} = -i \h(t) \ket{\psi_t}$, the equation of motion can be written as
\
\begin{equation}
\begin{aligned}\label{eq:EOM}
    \dot{r}(t) &= \frac{\epsilon(t)}{2} \sin \varphi(t),\\
    \dot{\varphi}(t) &= 2\omega -\epsilon(t) (1 -\coth 2r(t) \cos \varphi(t)).
\end{aligned}
\end{equation}
We then obtain the explicit form of QFI from Eq.~\eqref{eq:QFI_def},
\begin{equation}\label{eq:QFI_dynamical_encoding}
    {\cal F}_\omega=2 \left|\int_0^T dt \sinh 2r (t) e^{i\theta(t)}\right|^2,
\end{equation}
where $\theta(t) = -\int_0^t  \frac{\epsilon(t')\cos \varphi(t')}{\sinh 2r(t')} dt'$ (see Appendix~\ref{section:dynamic} for detailed derivation).

\begin{figure*}[t]
\includegraphics[width=.95\linewidth]{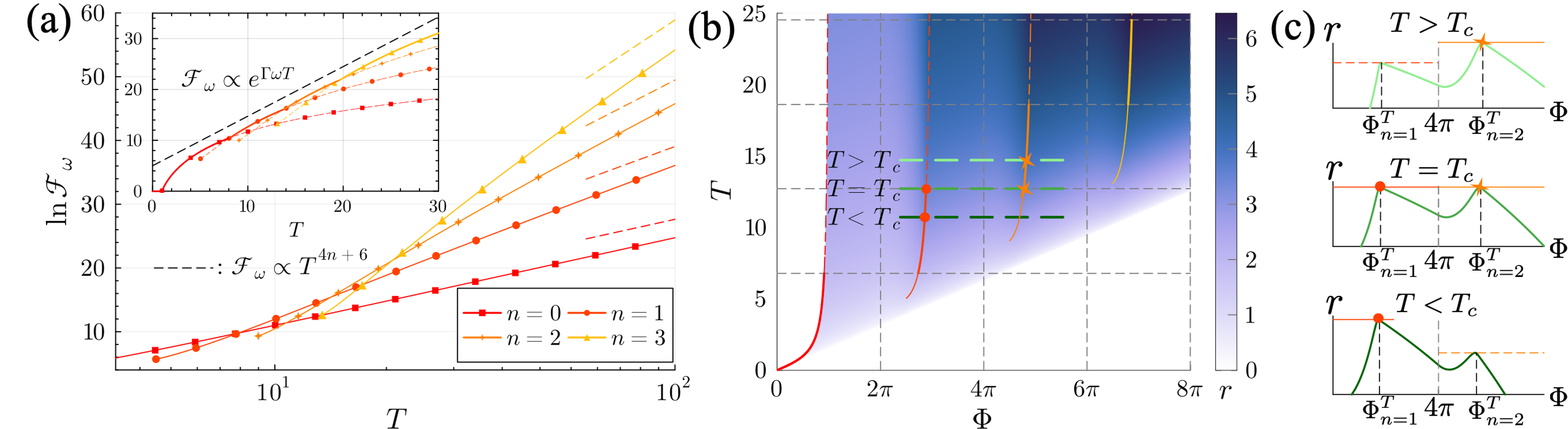}
\caption{(a) QFI for a given total evolution time $T$ by adapting the on-off control. The QFI is obtained by exact simulation using Eq.~\eqref{eq:EOM} and Eq.~\eqref{eq:QFI_dynamical_encoding}. The main figure verifies the $T^{4n+6}$ scaling at $T \gg 1$. The inset figure indicates the exponential scaling of QFI by taking the maximum QFI among different winding numbers. (b) The optimal squeezing parameter $r(T,\Phi)$ as a function of $T$ and the total phase $\Phi$. The squeezing parameter has a local maximum $r(T,\Phi_n^T)$ for each $n = \floor{\Phi/(2\pi)}$, indicated by the dashed lines.  The solid line indicates the global maximum $r_{\rm max}(T)$ for a given $T$, optimized over all possible $\Phi$. (c) The squeezing parameter and the optimal phase $\Phi_{\rm opt}^T$ near the critical time. When $T$ crosses the critical time $T_c$, the optimal phase undergoes a discontinuous jump as the optimal winding number changes from $n=1$ to $n=2$. At the critical time $T = T_c$, the optimal squeezing is given by the two winding numbers.
}\label{Fig:phase_transition}
\end{figure*}

\emph{QFI scaling in terms of the winding number.---} We show that the total phase accumulated during the dynamics, $\Phi =\int_0^T \dot{\varphi}(t) dt$, determines the scaling limit of the QFI. We introduce the winding number
\begin{equation}
n \equiv \left\lfloor \frac{\Phi}{2\pi} \right\rfloor,
\end{equation}
a topological quantity that counts the number of cycles completed throughout the entire phase-space trajectory, which is independent of the detailed shape of the trajectory (see Fig.~\ref{Fig:Concept}).  

We present a theorem that establishes the polynomial scaling bound of QFI when the winding number is fixed.
\begin{theorem} The scaling of QFI with a fixed winding number $n$ is bounded by
    \begin{equation}
    \mathcal{F}_\omega(T) \leq c_n T^{4n+6} + o(T^{4n+6}),
    \end{equation}
where $c_n$ is independent of $T$ and $o(T^{4n+6})$ grows much slower than $T^{4n+6}$.
\label{theorem:2(n+1)} 
\end{theorem}
An explicit form of $c_n$ can be found in the Appendix~\ref{section:fintinte_winding}. This implies that the winding number $n$ determines the order of the polynomial scaling of the QFI with respect to $T$. Based on this, we demonstrate the following no-go theorem on critical quantum metrology with any monotonically increasing $\epsilon(t)$:
\begin{theorem}\label{theorem:mono_increasing_g}
    For any control with monotonically increasing $\epsilon(t)$, the QFI scaling cannot exceed $T^{6}$ as the winding number remains zero.
\end{theorem}
We note that the monotonically increasing $\epsilon(t)$ with a trivial winding number $n=0$, encompasses both adiabatic~\cite{Garbe.2019} and quenching processes~\cite{Chu.2021, Garbe_2022} widely studied in critical quantum metrology. Thus, Theorem~\ref{theorem:mono_increasing_g} demonstrates that the QFI in such monotonic protocols is fundamentally limited to $T^6$ scaling, irrespective of the rate at which $\epsilon(t)$ changes, while the sudden quench protocol~\cite{Chu.2021} saturates the optimal scaling.

Theorems~\ref{theorem:2(n+1)} and \ref{theorem:mono_increasing_g} establish a fundamental limitation on the QFI scaling when the winding number $n$ is fixed, regardless of the total evolution time $T$. Surprisingly, however, when the winding number $n$ increases with $T$, this limitation can be overcome. In this case, we demonstrate that the exponential scaling of the QFI is achievable with the following bound of the exponent:
\begin{theorem}
\label{theorem:exponetial_g=1} 
The scaling bound of the QFI for a long-time limit ($T \gg 1$) is given by
\begin{equation}
\mathcal{F}_\omega(T) \propto e^{\Gamma \omega T}
\end{equation}
with $\Gamma \approx 0.9745$. This bound can be saturated by taking the winding number $n \approx 0.169 \omega T$.
\end{theorem}
Theorem~\ref{theorem:exponetial_g=1} identifies that the exponential scaling of QFI stems from increasing the winding number with total evolution time. Compared to previous results utilizing periodic driving with a fixed time period~\cite{Garbe2022}, Theorem~\ref{theorem:exponetial_g=1} improves the exponent by approximately a factor of $7$. We also note that this represents the ultimate bound on the exponent, achievable by optimizing the winding number for a given total evolution time.

\emph{On-off control for the optimal QFI scaling.---} Having established the fundamental scaling limit of the QFI, it is important to identify a time-dependent control, $\epsilon(t)$, that realizes the optimal scaling. Remarkably, this can be achieved by a simple on-off scheme based on the phase of the quantum state.

As we are interested in the scaling behavior in the long-time limit $T \gg 1$, we consider a system with sufficiently large squeezing, $|r| \gg 1$ after a finite time. In this limit, Eq.~\eqref{eq:EOM} becomes
\begin{equation}
\begin{aligned}\label{eq:EOM_large_r}
	\dot{r}(t) & \approx  \epsilon(t) \sin (\varphi/2)\cos(\varphi/2),\\
	\dot{\varphi}(t) &\approx 2\omega - 2\epsilon(t) \sin^2 (\varphi/2),
\end{aligned}
\end{equation}
and $\theta(t)$ becomes constant over time. The QFI in Eq.~\eqref{eq:QFI_dynamical_encoding} is then simplified as ${\cal F}_\omega \approx 2\left(\int_0^T dt \sinh(2r(t))\right)^2$, solely governed by squeezing.

We note that $\varphi(t)$ is independent of $r$ and monotonically increasing by $t$ as $\dot \varphi (t) \geq 0$. This allows us to express $\epsilon(t) = \frac{2\omega - \dot\varphi}{2\sin^2(\varphi/2)}$, which leads to an integral expression of the squeezing as
\begin{equation}\label{eq:integral_without_epsilon}
r(T, \Phi) = \int_0^{\Phi} d\varphi \left( \frac{2\omega}{\dot\varphi}-1  \right) \frac{\cot(\varphi/2)}{2},
\end{equation}
with two constraints of the total evolution time $T$ and the total phase $\Phi = \int_0^T \dot\varphi(t) dt$. The control parameter's range $0\leq \epsilon(t)\leq \omega$ yields the boundary condition $\frac{1}{2}\leq\tfrac{\omega}{\dot{\varphi}}\leq \frac{\sec^2(\varphi/2)}{2}$.

Let us first consider the case with a fixed winding number $n$. We show that there is a unique local maximum of $r(T, \Phi)$ within the region $n\pi \leq \Phi < (n+1)\pi$. The corresponding optimal protocol is given by the following on-off scheduling (see Appendix~\ref{section:fintinte_winding} for more details),
\begin{equation}
\frac{\epsilon}{\omega} = 
\begin{cases}
1 & (\cot (\varphi/2)  \geq \cot(\phi_n^T/2) )\\
0 & (\cot  (\varphi/2)  < \cot(\phi_n^T/2)),
\end{cases}
\end{equation}
where $\phi_n^T$ is determined by $T$ and the winding number $n$. The optimal squeezing behaves as
\begin{equation}
r(T,\Phi_n^T) \equiv \max_{n\pi \leq \Phi \leq (n+1)\pi} r(T,\Phi)  \approx (n+1)\ln (\omega T)+c_n,
\end{equation}
with some constant $c_n$. Here, $\Phi_n^T$ is the total phase within $n \leq \frac{\Phi}{2\pi} < n+1$ that gives the optimal squeezing. This achieves to the optimal scaling in Theorem~\ref{theorem:2(n+1)} as $\mathcal{F}_{\omega}\approx 4 \left| \int_0^T dt\sinh^2(2r)\right|\propto T^{4n+6}$~(see Appendix~\ref{section:fintinte_winding} for more details).

Next, we find the global optimum of squeezing $r(T,\Phi)$ for a given total evolution time $T$, by maximizing over all possible winding numbers,
\begin{equation}
r_\mathrm{max}(T) = \max_{n}r(T,\Phi_n^T) = r(T,\Phi^T_{\rm opt})\approx\frac{\Gamma}{4}\omega T,
\end{equation}
where $\Phi_{\rm opt}^T$ is the optimal phase among $\{ \Phi_n^T\}$ that gives the optimal squeezing. As $r_{\rm max}$ linearly scales with $T$, the QFI $\mathcal{F}_{\omega}\approx 2 \left| \int_0^T dt\sinh^2 \left(\frac{\Gamma}{2}\omega T \right)\right| \propto e^{\Gamma \omega T}$ scales exponentially on $T$, achieving the optimal scaling in Theorem~\ref{theorem:exponetial_g=1}. We also note that the optimal winding number linearly increases as $n^{T}_\mathrm{opt} \approx 0.169 \omega T$ for a sufficiently large $T$. The numerical simulation in Fig.~\ref{Fig:phase_transition}(a) verifies that the on-off controls yield the optimal scaling of QFI for both cases with a fixed winding number (Theorem~\ref{theorem:2(n+1)}) and with an increasing winding number by total evolution time (Theorem~\ref{theorem:exponetial_g=1}).

Interestingly, upon increasing $T$, we observe a series of abrupt jumps in the optimal phase $\Phi_{\rm opt}^T$, accompanied with a unit increase in the optimal winding number. At each transition point, the maximum squeezing $r_{\rm max}(T)$ changes its functional form, from $\propto (n+1) \ln(\omega T)$ to $\propto (n+2) \ln(\omega T)$. This discontinuity associated with the optimal winding number is a consequence of a first-order phase transition in the squeezing parameter, as illustrated in Fig.~\ref{Fig:phase_transition}(b) and (c).  This reflects the underlying topological nature in the optimal phase-space trajectory, marked by a change in the topological winding number.

\emph{Scaling of the QFI away from the critical point.---} We further explore the role of criticality in QFI scaling by considering a scenario where the system cannot reach the critical point. For a fixed winding number, the squeezing parameter $r$ saturates at a finite maximum value, which is in stark contrast to the critical case (Theorem~\ref{theorem:2(n+1)}), where $r$ continues to increase with the evolution time:
\begin{theorem}\label{theorem:fixed_n_withouT_criticality}
    For $0\leq \epsilon(t) \leq \epsilon_{\rm max} < \omega$ and a fixed winding number $n$, the squeezing $r$ is upper bounded as
    \begin{equation}
        \sinh 2r(T)  \leq 
        \frac{1}{(1-(\epsilon_{\rm max}/\omega))^{n+1}},
    \end{equation}
regardless of the total evolution time $T$. This implies that the QFI eventually saturates to $T^2$.
\end{theorem}
Nevertheless, the exponential scaling of QFI can still be achieved by allowing the winding number to increase with the total evolution time, which leads to the following generalization of Theorem~\ref{theorem:exponetial_g=1}:
\begin{theorem}\label{theorem:expon_g<1}
With the control parameter in the range $0 \leq \epsilon(t) \leq \epsilon_\mathrm{max}$, the fundamental scaling limit is given as
\begin{equation}
\mathcal{F}_\omega(T) \propto e^{\Gamma(\epsilon_\mathrm{max}) \omega T}.
\end{equation}
\end{theorem}
We highlight that $\Gamma(\omega) = \Gamma$ described in Theorem~\ref{theorem:exponetial_g=1}, and  $\Gamma(\epsilon_\mathrm{max})$ \emph{smoothly} decreases with decreasing $\epsilon_\mathrm{max}$ and does not vanish even when $\epsilon_\mathrm{max}$ is far from the critical point $\omega$ (see Fig.~\ref{FIG:exp_without_critical_and_under_disspation}(a)). This observation has interesting implications for designing the optimal protocol to enhance QFI, suggesting that implementing a non-monotonic control to increase the winding number is a more crucial factor than driving the system near the critical point. Since the energy gap opening due to the finite-size effect plays a role analogous to moving away from the critical point in critical dynamics~\cite{Hwang.2015}, our analysis is expected to remain valid for capturing the scaling behavior of the QFI beyond the thermodynamic limit.

\begin{figure}[t]
\begin{center}
\includegraphics[width=.85\linewidth]{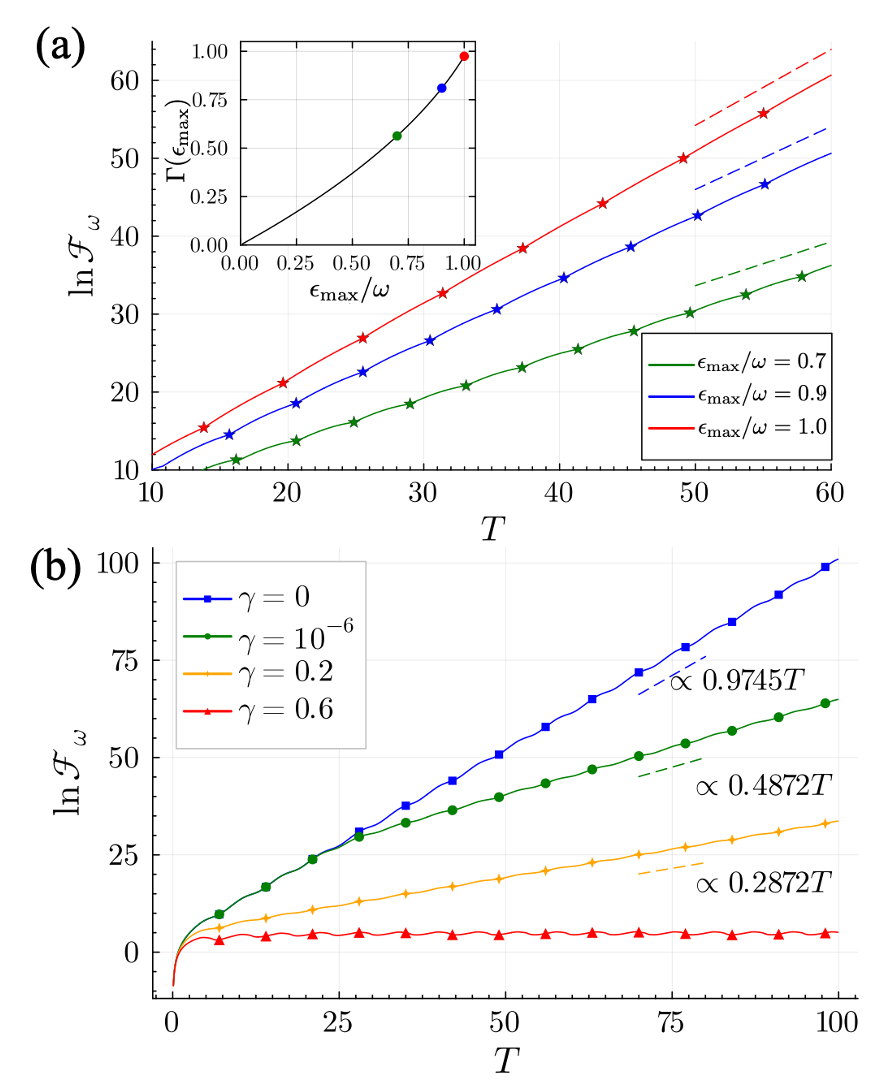}
\caption{(a) Exponential scaling of QFI with $\epsilon_{\rm max}/\omega= 0.7$, $0.9$, and $1.0$. The star marker indicates the critical time at which the winding number changes. The dashed line indicates the exponential scaling with $\Gamma(\epsilon_{\rm max})$. The inset figure describes the exponent $\Gamma(\epsilon_{\rm max})$. (b) Exponential scaling of QFI with various thermalization rate $\gamma = 0$, $0.2$, $0.6$, and $10^{-6}$. The dash lines guide the exponential scaling of QFI as $\mathcal{F}_\omega \propto e^{(\Gamma/2 -\gamma)T}$. For both cases, the QFI is obtained by the optimized on-off control.
}\label{FIG:exp_without_critical_and_under_disspation}
\end{center}
\end{figure}

\emph{Robustness under thermal dissipation.---}
We analyze the scaling limit of QFI under thermal dissipation by considering the following Lindblad equation,
\begin{equation}\label{eq:Lindblad_eq}
\begin{split}
\frac{d}{dt}\R=& - i[\h(t),\R] + \gamma (\bar{n}+1) (\an \R \an^\dagger -\tfrac{1}{2}\{\an^\dagger \an, \R \})\\
&+ \gamma \bar{n} (\an^\dagger \R \an -\tfrac{1}{2}\{\an \an^\dagger, \R \}),
\end{split}
\end{equation}
where $\gamma$ is the thermalization rate, $\bar{n}$ is the mean occupation number of the bath, and $\{ A, B\}= AB + BA$.

Remarkably, the exponential scaling persists under thermal dissipation but with a decreasing exponent~(see Appendix~\ref{section:exponential} for more details),
    \begin{equation}
    \mathcal{F}_\omega(T) \propto e^{ \left( \frac{\Gamma(\epsilon_{\rm max})  \omega }{2} - \gamma 
    \right)T}.
    \end{equation}
The exponent decreases with increasing thermalization rate $\gamma$ but is not affected by the mean occupation number $\bar{n}$. The factor $1/2$ in the exponent arises from the long-time behavior under thermal dissipation $T \gg \ln(1/\gamma)$. For effectively closed dynamics within the time window $T \lesssim \ln(1/\gamma)$, the exponent is given by $\Gamma(\epsilon_{\rm max})$. These two different behaviors can be confirmed through numerical simulation in Fig.~\ref{FIG:exp_without_critical_and_under_disspation}(b).

\emph{Remarks.---}
We have established the fundamental scaling limit of critical quantum metrology for a given total evolution time $T$. We have demonstrated that when the winding number $n$ in phase space trajectory is fixed, the QFI scaling is upper bounded by ${\cal F}_\omega(T) \propto T^{4n+6}$. Our work reveals the role of the topological properties of phase space trajectories as a unifying principle for maximizing accuracy, applicable across various critical metrology scenarios explored in previous studies. Furthermore, we have shown that an exponential scaling of the QFI, ${\cal F}_\omega(T) \propto e^{\Gamma (\omega T)}$, can be achieved by increasing the winding number with total evolution time and have derived the fundamental bound of the exponent $\Gamma$.

We have also developed an explicit protocol to achieve optimal scaling of QFI. This protocol can be implemented by simple on-off scheduling of $\epsilon(t)$, thereby providing an experimentally feasible approach to high-precision metrology through a simple dynamic control of the quantum system. We emphasize that such dynamic control becomes even more crucial in realistic scenarios where reaching the critical point is not feasible. In such cases, as we have shown, a monotonically increasing $\epsilon(t)$ is limited to achieving only quadratic scaling of the QFI; hence, non-monotonic control, such as on-off scheduling, is essential to attain exponential scaling. We have also shown the robustness of the exponential scaling of QFI under thermal dissipation.

The theoretical framework in our work offers a potential avenue for exploring optimal control beyond critical quantum metrology, for example, enhancing squeezing~\cite{Abah2022} or charging quantum batteries~\cite{Abah2022, Downing2024}. An intriguing future research direction would be extending our formalism beyond the Gaussian limit. Another interesting question is whether the periodic control studied in this work could provide a promising route to circumvent critical slowing down in quantum many-body dynamics.

\begin{acknowledgements}
This work was supported by the National Research Foundation of Korea (NRF) grant funded by the Korea government (MSIT) (No. RS-2024-00413957 and No. RS-2024-00438415). HK is supported by the KIAS Individual Grant No. CG085302 at Korea Institute for Advanced Study. MJH is supported by the Innovation Program for Quantum Science and Technology 2021ZD0301602.
\end{acknowledgements}

\bibliography{bib.bib}

\begin{thebibliography}{52}%
\makeatletter
\providecommand \@ifxundefined [1]{%
 \@ifx{#1\undefined}
}%
\providecommand \@ifnum [1]{%
 \ifnum #1\expandafter \@firstoftwo
 \else \expandafter \@secondoftwo
 \fi
}%
\providecommand \@ifx [1]{%
 \ifx #1\expandafter \@firstoftwo
 \else \expandafter \@secondoftwo
 \fi
}%
\providecommand \natexlab [1]{#1}%
\providecommand \enquote  [1]{``#1''}%
\providecommand \bibnamefont  [1]{#1}%
\providecommand \bibfnamefont [1]{#1}%
\providecommand \citenamefont [1]{#1}%
\providecommand \href@noop [0]{\@secondoftwo}%
\providecommand \href [0]{\begingroup \@sanitize@url \@href}%
\providecommand \@href[1]{\@@startlink{#1}\@@href}%
\providecommand \@@href[1]{\endgroup#1\@@endlink}%
\providecommand \@sanitize@url [0]{\catcode `\\12\catcode `\$12\catcode
  `\&12\catcode `\#12\catcode `\^12\catcode `\_12\catcode `\%12\relax}%
\providecommand \@@startlink[1]{}%
\providecommand \@@endlink[0]{}%
\providecommand \url  [0]{\begingroup\@sanitize@url \@url }%
\providecommand \@url [1]{\endgroup\@href {#1}{\urlprefix }}%
\providecommand \urlprefix  [0]{URL }%
\providecommand \Eprint [0]{\href }%
\providecommand \doibase [0]{https://doi.org/}%
\providecommand \selectlanguage [0]{\@gobble}%
\providecommand \bibinfo  [0]{\@secondoftwo}%
\providecommand \bibfield  [0]{\@secondoftwo}%
\providecommand \translation [1]{[#1]}%
\providecommand \BibitemOpen [0]{}%
\providecommand \bibitemStop [0]{}%
\providecommand \bibitemNoStop [0]{.\EOS\space}%
\providecommand \EOS [0]{\spacefactor3000\relax}%
\providecommand \BibitemShut  [1]{\csname bibitem#1\endcsname}%
\let\auto@bib@innerbib\@empty
\bibitem [{\citenamefont {Giovannetti}\ \emph {et~al.}(2006)\citenamefont
  {Giovannetti}, \citenamefont {Lloyd},\ and\ \citenamefont
  {Maccone}}]{Giovannetti.2006}%
  \BibitemOpen
  \bibfield  {author} {\bibinfo {author} {\bibfnamefont {V.}~\bibnamefont
  {Giovannetti}}, \bibinfo {author} {\bibfnamefont {S.}~\bibnamefont {Lloyd}},\
  and\ \bibinfo {author} {\bibfnamefont {L.}~\bibnamefont {Maccone}},\
  }\bibfield  {title} {\bibinfo {title} {Quantum metrology},\ }\href
  {https://doi.org/10.1103/PhysRevLett.96.010401} {\bibfield  {journal}
  {\bibinfo  {journal} {Phys. Rev. Lett.}\ }\textbf {\bibinfo {volume} {96}},\
  \bibinfo {pages} {010401} (\bibinfo {year} {2006})}\BibitemShut {NoStop}%
\bibitem [{\citenamefont {Giovannetti}\ \emph {et~al.}(2011)\citenamefont
  {Giovannetti}, \citenamefont {Lloyd},\ and\ \citenamefont
  {Maccone}}]{Giovannetti.2011}%
  \BibitemOpen
  \bibfield  {author} {\bibinfo {author} {\bibfnamefont {V.}~\bibnamefont
  {Giovannetti}}, \bibinfo {author} {\bibfnamefont {S.}~\bibnamefont {Lloyd}},\
  and\ \bibinfo {author} {\bibfnamefont {L.}~\bibnamefont {Maccone}},\
  }\bibfield  {title} {\bibinfo {title} {Advances in quantum metrology},\
  }\href@noop {} {\bibfield  {journal} {\bibinfo  {journal} {Nat. Photonics}\
  }\textbf {\bibinfo {volume} {5}},\ \bibinfo {pages} {222} (\bibinfo {year}
  {2011})}\BibitemShut {NoStop}%
\bibitem [{\citenamefont {Chu}\ \emph {et~al.}(2023)\citenamefont {Chu},
  \citenamefont {Li},\ and\ \citenamefont {Cai}}]{Chu2023}%
  \BibitemOpen
  \bibfield  {author} {\bibinfo {author} {\bibfnamefont {Y.}~\bibnamefont
  {Chu}}, \bibinfo {author} {\bibfnamefont {X.}~\bibnamefont {Li}},\ and\
  \bibinfo {author} {\bibfnamefont {J.}~\bibnamefont {Cai}},\ }\bibfield
  {title} {\bibinfo {title} {Strong quantum metrological limit from many-body
  physics},\ }\href {https://doi.org/10.1103/PhysRevLett.130.170801} {\bibfield
   {journal} {\bibinfo  {journal} {Phys. Rev. Lett.}\ }\textbf {\bibinfo
  {volume} {130}},\ \bibinfo {pages} {170801} (\bibinfo {year}
  {2023})}\BibitemShut {NoStop}%
\bibitem [{\citenamefont {Schnabel}\ \emph {et~al.}(2010)\citenamefont
  {Schnabel}, \citenamefont {Mavalvala}, \citenamefont {McClelland},\ and\
  \citenamefont {Lam}}]{Schnabel10}%
  \BibitemOpen
  \bibfield  {author} {\bibinfo {author} {\bibfnamefont {R.}~\bibnamefont
  {Schnabel}}, \bibinfo {author} {\bibfnamefont {N.}~\bibnamefont {Mavalvala}},
  \bibinfo {author} {\bibfnamefont {D.~E.}\ \bibnamefont {McClelland}},\ and\
  \bibinfo {author} {\bibfnamefont {P.~K.}\ \bibnamefont {Lam}},\ }\bibfield
  {title} {\bibinfo {title} {Quantum metrology for gravitational wave
  astronomy},\ }\href@noop {} {\bibfield  {journal} {\bibinfo  {journal} {Nat.
  Commun.}\ }\textbf {\bibinfo {volume} {1}},\ \bibinfo {pages} {121} (\bibinfo
  {year} {2010})}\BibitemShut {NoStop}%
\bibitem [{\citenamefont {Aslam}\ \emph {et~al.}(2023)\citenamefont {Aslam},
  \citenamefont {Zhou}, \citenamefont {Urbach}, \citenamefont {Turner},
  \citenamefont {Walsworth}, \citenamefont {Lukin},\ and\ \citenamefont
  {Park}}]{Aslam23}%
  \BibitemOpen
  \bibfield  {author} {\bibinfo {author} {\bibfnamefont {N.}~\bibnamefont
  {Aslam}}, \bibinfo {author} {\bibfnamefont {H.}~\bibnamefont {Zhou}},
  \bibinfo {author} {\bibfnamefont {E.~K.}\ \bibnamefont {Urbach}}, \bibinfo
  {author} {\bibfnamefont {M.~J.}\ \bibnamefont {Turner}}, \bibinfo {author}
  {\bibfnamefont {R.~L.}\ \bibnamefont {Walsworth}}, \bibinfo {author}
  {\bibfnamefont {M.~D.}\ \bibnamefont {Lukin}},\ and\ \bibinfo {author}
  {\bibfnamefont {H.}~\bibnamefont {Park}},\ }\bibfield  {title} {\bibinfo
  {title} {Quantum sensors for biomedical applications},\ }\href@noop {}
  {\bibfield  {journal} {\bibinfo  {journal} {Nat. Rev. Phys.}\ }\textbf
  {\bibinfo {volume} {5}},\ \bibinfo {pages} {157} (\bibinfo {year}
  {2023})}\BibitemShut {NoStop}%
\bibitem [{\citenamefont {Maccone}\ and\ \citenamefont
  {Ren}(2020)}]{Maccone20}%
  \BibitemOpen
  \bibfield  {author} {\bibinfo {author} {\bibfnamefont {L.}~\bibnamefont
  {Maccone}}\ and\ \bibinfo {author} {\bibfnamefont {C.}~\bibnamefont {Ren}},\
  }\bibfield  {title} {\bibinfo {title} {Quantum radar},\ }\href
  {https://doi.org/10.1103/PhysRevLett.124.200503} {\bibfield  {journal}
  {\bibinfo  {journal} {Phys. Rev. Lett.}\ }\textbf {\bibinfo {volume} {124}},\
  \bibinfo {pages} {200503} (\bibinfo {year} {2020})}\BibitemShut {NoStop}%
\bibitem [{\citenamefont {Couteau}\ \emph {et~al.}(2023)\citenamefont
  {Couteau}, \citenamefont {Barz}, \citenamefont {Durt}, \citenamefont
  {Gerrits}, \citenamefont {Huwer}, \citenamefont {Prevedel}, \citenamefont
  {Rarity}, \citenamefont {Shields},\ and\ \citenamefont {Weihs}}]{Couteau23}%
  \BibitemOpen
  \bibfield  {author} {\bibinfo {author} {\bibfnamefont {C.}~\bibnamefont
  {Couteau}}, \bibinfo {author} {\bibfnamefont {S.}~\bibnamefont {Barz}},
  \bibinfo {author} {\bibfnamefont {T.}~\bibnamefont {Durt}}, \bibinfo {author}
  {\bibfnamefont {T.}~\bibnamefont {Gerrits}}, \bibinfo {author} {\bibfnamefont
  {J.}~\bibnamefont {Huwer}}, \bibinfo {author} {\bibfnamefont
  {R.}~\bibnamefont {Prevedel}}, \bibinfo {author} {\bibfnamefont
  {J.}~\bibnamefont {Rarity}}, \bibinfo {author} {\bibfnamefont
  {A.}~\bibnamefont {Shields}},\ and\ \bibinfo {author} {\bibfnamefont
  {G.}~\bibnamefont {Weihs}},\ }\bibfield  {title} {\bibinfo {title}
  {Applications of single photons in quantum metrology, biology and the
  foundations of quantum physics},\ }\href@noop {} {\bibfield  {journal}
  {\bibinfo  {journal} {Nat. Rev. Phys.}\ }\textbf {\bibinfo {volume} {5}},\
  \bibinfo {pages} {354} (\bibinfo {year} {2023})}\BibitemShut {NoStop}%
\bibitem [{\citenamefont {Pezz\`e}\ \emph {et~al.}(2018)\citenamefont
  {Pezz\`e}, \citenamefont {Smerzi}, \citenamefont {Oberthaler}, \citenamefont
  {Schmied},\ and\ \citenamefont {Treutlein}}]{Pezze.2018}%
  \BibitemOpen
  \bibfield  {author} {\bibinfo {author} {\bibfnamefont {L.}~\bibnamefont
  {Pezz\`e}}, \bibinfo {author} {\bibfnamefont {A.}~\bibnamefont {Smerzi}},
  \bibinfo {author} {\bibfnamefont {M.~K.}\ \bibnamefont {Oberthaler}},
  \bibinfo {author} {\bibfnamefont {R.}~\bibnamefont {Schmied}},\ and\ \bibinfo
  {author} {\bibfnamefont {P.}~\bibnamefont {Treutlein}},\ }\bibfield  {title}
  {\bibinfo {title} {Quantum metrology with nonclassical states of atomic
  ensembles},\ }\href {https://doi.org/10.1103/RevModPhys.90.035005} {\bibfield
   {journal} {\bibinfo  {journal} {Rev. Mod. Phys.}\ }\textbf {\bibinfo
  {volume} {90}},\ \bibinfo {pages} {035005} (\bibinfo {year}
  {2018})}\BibitemShut {NoStop}%
\bibitem [{\citenamefont {Kwon}\ \emph {et~al.}(2019)\citenamefont {Kwon},
  \citenamefont {Tan}, \citenamefont {Volkoff},\ and\ \citenamefont
  {Jeong}}]{Kwon.2019}%
  \BibitemOpen
  \bibfield  {author} {\bibinfo {author} {\bibfnamefont {H.}~\bibnamefont
  {Kwon}}, \bibinfo {author} {\bibfnamefont {K.~C.}\ \bibnamefont {Tan}},
  \bibinfo {author} {\bibfnamefont {T.}~\bibnamefont {Volkoff}},\ and\ \bibinfo
  {author} {\bibfnamefont {H.}~\bibnamefont {Jeong}},\ }\bibfield  {title}
  {\bibinfo {title} {Nonclassicality as a quantifiable resource for quantum
  metrology},\ }\href {https://doi.org/10.1103/PhysRevLett.122.040503}
  {\bibfield  {journal} {\bibinfo  {journal} {Phys. Rev. Lett.}\ }\textbf
  {\bibinfo {volume} {122}},\ \bibinfo {pages} {040503} (\bibinfo {year}
  {2019})}\BibitemShut {NoStop}%
\bibitem [{\citenamefont {Ge}\ \emph {et~al.}(2020)\citenamefont {Ge},
  \citenamefont {Jacobs}, \citenamefont {Asiri}, \citenamefont {Foss-Feig},\
  and\ \citenamefont {Zubairy}}]{Ge20}%
  \BibitemOpen
  \bibfield  {author} {\bibinfo {author} {\bibfnamefont {W.}~\bibnamefont
  {Ge}}, \bibinfo {author} {\bibfnamefont {K.}~\bibnamefont {Jacobs}}, \bibinfo
  {author} {\bibfnamefont {S.}~\bibnamefont {Asiri}}, \bibinfo {author}
  {\bibfnamefont {M.}~\bibnamefont {Foss-Feig}},\ and\ \bibinfo {author}
  {\bibfnamefont {M.~S.}\ \bibnamefont {Zubairy}},\ }\bibfield  {title}
  {\bibinfo {title} {Operational resource theory of nonclassicality via quantum
  metrology},\ }\href {https://doi.org/10.1103/PhysRevResearch.2.023400}
  {\bibfield  {journal} {\bibinfo  {journal} {Phys. Rev. Res.}\ }\textbf
  {\bibinfo {volume} {2}},\ \bibinfo {pages} {023400} (\bibinfo {year}
  {2020})}\BibitemShut {NoStop}%
\bibitem [{\citenamefont {Maccone}\ and\ \citenamefont
  {Riccardi}(2020)}]{Maccone2020squeezingmetrology}%
  \BibitemOpen
  \bibfield  {author} {\bibinfo {author} {\bibfnamefont {L.}~\bibnamefont
  {Maccone}}\ and\ \bibinfo {author} {\bibfnamefont {A.}~\bibnamefont
  {Riccardi}},\ }\bibfield  {title} {\bibinfo {title} {Squeezing metrology: a
  unified framework},\ }\href {https://doi.org/10.22331/q-2020-07-09-292}
  {\bibfield  {journal} {\bibinfo  {journal} {{Quantum}}\ }\textbf {\bibinfo
  {volume} {4}},\ \bibinfo {pages} {292} (\bibinfo {year} {2020})}\BibitemShut
  {NoStop}%
\bibitem [{\citenamefont {Braunstein}\ and\ \citenamefont
  {Caves}(1994)}]{Braunstein.1994}%
  \BibitemOpen
  \bibfield  {author} {\bibinfo {author} {\bibfnamefont {S.~L.}\ \bibnamefont
  {Braunstein}}\ and\ \bibinfo {author} {\bibfnamefont {C.~M.}\ \bibnamefont
  {Caves}},\ }\bibfield  {title} {\bibinfo {title} {Statistical distance and
  the geometry of quantum states},\ }\href
  {https://doi.org/10.1103/PhysRevLett.72.3439} {\bibfield  {journal} {\bibinfo
   {journal} {Phys. Rev. Lett.}\ }\textbf {\bibinfo {volume} {72}},\ \bibinfo
  {pages} {3439} (\bibinfo {year} {1994})}\BibitemShut {NoStop}%
\bibitem [{\citenamefont {Mitchell}\ \emph {et~al.}(2004)\citenamefont
  {Mitchell}, \citenamefont {Lundeen},\ and\ \citenamefont
  {Steinberg}}]{Mitchell04}%
  \BibitemOpen
  \bibfield  {author} {\bibinfo {author} {\bibfnamefont {M.~W.}\ \bibnamefont
  {Mitchell}}, \bibinfo {author} {\bibfnamefont {J.~S.}\ \bibnamefont
  {Lundeen}},\ and\ \bibinfo {author} {\bibfnamefont {A.~M.}\ \bibnamefont
  {Steinberg}},\ }\bibfield  {title} {\bibinfo {title} {Super-resolving phase
  measurements with a multiphoton entangled state},\ }\href@noop {} {\bibfield
  {journal} {\bibinfo  {journal} {Nature}\ }\textbf {\bibinfo {volume} {429}},\
  \bibinfo {pages} {161} (\bibinfo {year} {2004})}\BibitemShut {NoStop}%
\bibitem [{\citenamefont {Leibfried}\ \emph {et~al.}(2004)\citenamefont
  {Leibfried}, \citenamefont {Barrett}, \citenamefont {Schaetz}, \citenamefont
  {Britton}, \citenamefont {Chiaverini}, \citenamefont {Itano}, \citenamefont
  {Jost}, \citenamefont {Langer},\ and\ \citenamefont
  {Wineland}}]{Leibfried04}%
  \BibitemOpen
  \bibfield  {author} {\bibinfo {author} {\bibfnamefont {D.}~\bibnamefont
  {Leibfried}}, \bibinfo {author} {\bibfnamefont {M.~D.}\ \bibnamefont
  {Barrett}}, \bibinfo {author} {\bibfnamefont {T.}~\bibnamefont {Schaetz}},
  \bibinfo {author} {\bibfnamefont {J.}~\bibnamefont {Britton}}, \bibinfo
  {author} {\bibfnamefont {J.}~\bibnamefont {Chiaverini}}, \bibinfo {author}
  {\bibfnamefont {W.~M.}\ \bibnamefont {Itano}}, \bibinfo {author}
  {\bibfnamefont {J.~D.}\ \bibnamefont {Jost}}, \bibinfo {author}
  {\bibfnamefont {C.}~\bibnamefont {Langer}},\ and\ \bibinfo {author}
  {\bibfnamefont {D.~J.}\ \bibnamefont {Wineland}},\ }\bibfield  {title}
  {\bibinfo {title} {Toward heisenberg-limited spectroscopy with multiparticle
  entangled states},\ }\href@noop {} {\bibfield  {journal} {\bibinfo  {journal}
  {Science}\ }\textbf {\bibinfo {volume} {304}},\ \bibinfo {pages} {1476}
  (\bibinfo {year} {2004})}\BibitemShut {NoStop}%
\bibitem [{\citenamefont {Riedel}\ \emph {et~al.}(2010)\citenamefont {Riedel},
  \citenamefont {B{\"o}hi}, \citenamefont {Li}, \citenamefont {H{\"a}nsch},
  \citenamefont {Sinatra},\ and\ \citenamefont {Treutlein}}]{Riedel10}%
  \BibitemOpen
  \bibfield  {author} {\bibinfo {author} {\bibfnamefont {M.~F.}\ \bibnamefont
  {Riedel}}, \bibinfo {author} {\bibfnamefont {P.}~\bibnamefont {B{\"o}hi}},
  \bibinfo {author} {\bibfnamefont {Y.}~\bibnamefont {Li}}, \bibinfo {author}
  {\bibfnamefont {T.~W.}\ \bibnamefont {H{\"a}nsch}}, \bibinfo {author}
  {\bibfnamefont {A.}~\bibnamefont {Sinatra}},\ and\ \bibinfo {author}
  {\bibfnamefont {P.}~\bibnamefont {Treutlein}},\ }\bibfield  {title} {\bibinfo
  {title} {Atom-chip-based generation of entanglement for quantum metrology},\
  }\href@noop {} {\bibfield  {journal} {\bibinfo  {journal} {Nature}\ }\textbf
  {\bibinfo {volume} {464}},\ \bibinfo {pages} {1170} (\bibinfo {year}
  {2010})}\BibitemShut {NoStop}%
\bibitem [{\citenamefont {Gross}\ \emph {et~al.}(2010)\citenamefont {Gross},
  \citenamefont {Zibold}, \citenamefont {Nicklas}, \citenamefont {Est{\`e}ve},\
  and\ \citenamefont {Oberthaler}}]{Gross10}%
  \BibitemOpen
  \bibfield  {author} {\bibinfo {author} {\bibfnamefont {C.}~\bibnamefont
  {Gross}}, \bibinfo {author} {\bibfnamefont {T.}~\bibnamefont {Zibold}},
  \bibinfo {author} {\bibfnamefont {E.}~\bibnamefont {Nicklas}}, \bibinfo
  {author} {\bibfnamefont {J.}~\bibnamefont {Est{\`e}ve}},\ and\ \bibinfo
  {author} {\bibfnamefont {M.~K.}\ \bibnamefont {Oberthaler}},\ }\bibfield
  {title} {\bibinfo {title} {Nonlinear atom interferometer surpasses classical
  precision limit},\ }\href@noop {} {\bibfield  {journal} {\bibinfo  {journal}
  {Nature}\ }\textbf {\bibinfo {volume} {464}},\ \bibinfo {pages} {1165}
  (\bibinfo {year} {2010})}\BibitemShut {NoStop}%
\bibitem [{\citenamefont {Monz}\ \emph {et~al.}(2011)\citenamefont {Monz},
  \citenamefont {Schindler}, \citenamefont {Barreiro}, \citenamefont {Chwalla},
  \citenamefont {Nigg}, \citenamefont {Coish}, \citenamefont {Harlander},
  \citenamefont {H\"ansel}, \citenamefont {Hennrich},\ and\ \citenamefont
  {Blatt}}]{Monz11}%
  \BibitemOpen
  \bibfield  {author} {\bibinfo {author} {\bibfnamefont {T.}~\bibnamefont
  {Monz}}, \bibinfo {author} {\bibfnamefont {P.}~\bibnamefont {Schindler}},
  \bibinfo {author} {\bibfnamefont {J.~T.}\ \bibnamefont {Barreiro}}, \bibinfo
  {author} {\bibfnamefont {M.}~\bibnamefont {Chwalla}}, \bibinfo {author}
  {\bibfnamefont {D.}~\bibnamefont {Nigg}}, \bibinfo {author} {\bibfnamefont
  {W.~A.}\ \bibnamefont {Coish}}, \bibinfo {author} {\bibfnamefont
  {M.}~\bibnamefont {Harlander}}, \bibinfo {author} {\bibfnamefont
  {W.}~\bibnamefont {H\"ansel}}, \bibinfo {author} {\bibfnamefont
  {M.}~\bibnamefont {Hennrich}},\ and\ \bibinfo {author} {\bibfnamefont
  {R.}~\bibnamefont {Blatt}},\ }\bibfield  {title} {\bibinfo {title} {14-qubit
  entanglement: Creation and coherence},\ }\href
  {https://doi.org/10.1103/PhysRevLett.106.130506} {\bibfield  {journal}
  {\bibinfo  {journal} {Phys. Rev. Lett.}\ }\textbf {\bibinfo {volume} {106}},\
  \bibinfo {pages} {130506} (\bibinfo {year} {2011})}\BibitemShut {NoStop}%
\bibitem [{\citenamefont {Kandala}\ \emph {et~al.}(2019)\citenamefont
  {Kandala}, \citenamefont {Temme}, \citenamefont {C{\'o}rcoles}, \citenamefont
  {Mezzacapo}, \citenamefont {Chow},\ and\ \citenamefont
  {Gambetta}}]{Abhinav19}%
  \BibitemOpen
  \bibfield  {author} {\bibinfo {author} {\bibfnamefont {A.}~\bibnamefont
  {Kandala}}, \bibinfo {author} {\bibfnamefont {K.}~\bibnamefont {Temme}},
  \bibinfo {author} {\bibfnamefont {A.~D.}\ \bibnamefont {C{\'o}rcoles}},
  \bibinfo {author} {\bibfnamefont {A.}~\bibnamefont {Mezzacapo}}, \bibinfo
  {author} {\bibfnamefont {J.~M.}\ \bibnamefont {Chow}},\ and\ \bibinfo
  {author} {\bibfnamefont {J.~M.}\ \bibnamefont {Gambetta}},\ }\bibfield
  {title} {\bibinfo {title} {Error mitigation extends the computational reach
  of a noisy quantum processor},\ }\href@noop {} {\bibfield  {journal}
  {\bibinfo  {journal} {Nature}\ }\textbf {\bibinfo {volume} {567}},\ \bibinfo
  {pages} {491} (\bibinfo {year} {2019})}\BibitemShut {NoStop}%
\bibitem [{\citenamefont {Mu\~noz Arias}\ \emph {et~al.}(2023)\citenamefont
  {Mu\~noz Arias}, \citenamefont {Deutsch},\ and\ \citenamefont
  {Poggi}}]{MunozArias23}%
  \BibitemOpen
  \bibfield  {author} {\bibinfo {author} {\bibfnamefont {M.~H.}\ \bibnamefont
  {Mu\~noz Arias}}, \bibinfo {author} {\bibfnamefont {I.~H.}\ \bibnamefont
  {Deutsch}},\ and\ \bibinfo {author} {\bibfnamefont {P.~M.}\ \bibnamefont
  {Poggi}},\ }\bibfield  {title} {\bibinfo {title} {Phase-space geometry and
  optimal state preparation in quantum metrology with collective spins},\
  }\href {https://doi.org/10.1103/PRXQuantum.4.020314} {\bibfield  {journal}
  {\bibinfo  {journal} {PRX Quantum}\ }\textbf {\bibinfo {volume} {4}},\
  \bibinfo {pages} {020314} (\bibinfo {year} {2023})}\BibitemShut {NoStop}%
\bibitem [{\citenamefont {Zanardi}\ \emph {et~al.}(2008)\citenamefont
  {Zanardi}, \citenamefont {Paris},\ and\ \citenamefont
  {Venuti}}]{Zanardi.2008}%
  \BibitemOpen
  \bibfield  {author} {\bibinfo {author} {\bibfnamefont {P.}~\bibnamefont
  {Zanardi}}, \bibinfo {author} {\bibfnamefont {M.~G.~A.}\ \bibnamefont
  {Paris}},\ and\ \bibinfo {author} {\bibfnamefont {L.~C.}\ \bibnamefont
  {Venuti}},\ }\bibfield  {title} {\bibinfo {title} {{Quantum criticality as a
  resource for quantum estimation}},\ }\href
  {https://doi.org/10.1103/physreva.78.042105} {\bibfield  {journal} {\bibinfo
  {journal} {Phys. Rev. A}\ }\textbf {\bibinfo {volume} {78}},\ \bibinfo
  {pages} {042105} (\bibinfo {year} {2008})}\BibitemShut {NoStop}%
\bibitem [{\citenamefont {Macieszczak}\ \emph {et~al.}(2016)\citenamefont
  {Macieszczak}, \citenamefont {Guţă}, \citenamefont {Lesanovsky},\ and\
  \citenamefont {Garrahan}}]{Macieszczak.2016}%
  \BibitemOpen
  \bibfield  {author} {\bibinfo {author} {\bibfnamefont {K.}~\bibnamefont
  {Macieszczak}}, \bibinfo {author} {\bibfnamefont {M.}~\bibnamefont {Guţă}},
  \bibinfo {author} {\bibfnamefont {I.}~\bibnamefont {Lesanovsky}},\ and\
  \bibinfo {author} {\bibfnamefont {J.~P.}\ \bibnamefont {Garrahan}},\
  }\bibfield  {title} {\bibinfo {title} {{Dynamical phase transitions as a
  resource for quantum enhanced metrology}},\ }\href
  {https://doi.org/10.1103/physreva.93.022103} {\bibfield  {journal} {\bibinfo
  {journal} {Phys. Rev. A}\ }\textbf {\bibinfo {volume} {93}},\ \bibinfo
  {pages} {022103} (\bibinfo {year} {2016})}\BibitemShut {NoStop}%
\bibitem [{\citenamefont {Fernández-Lorenzo}\ and\ \citenamefont
  {Porras}(2017)}]{Fernandez-Lorenzo.2017}%
  \BibitemOpen
  \bibfield  {author} {\bibinfo {author} {\bibfnamefont {S.}~\bibnamefont
  {Fernández-Lorenzo}}\ and\ \bibinfo {author} {\bibfnamefont
  {D.}~\bibnamefont {Porras}},\ }\bibfield  {title} {\bibinfo {title} {{Quantum
  sensing close to a dissipative phase transition: Symmetry breaking and
  criticality as metrological resources}},\ }\href
  {https://doi.org/10.1103/physreva.96.013817} {\bibfield  {journal} {\bibinfo
  {journal} {Phys. Rev. A}\ }\textbf {\bibinfo {volume} {96}},\ \bibinfo
  {pages} {013817} (\bibinfo {year} {2017})}\BibitemShut {NoStop}%
\bibitem [{\citenamefont {Rams}\ \emph {et~al.}(2018)\citenamefont {Rams},
  \citenamefont {Sierant}, \citenamefont {Dutta}, \citenamefont {Horodecki},\
  and\ \citenamefont {Zakrzewski}}]{Rams.2018}%
  \BibitemOpen
  \bibfield  {author} {\bibinfo {author} {\bibfnamefont {M.~M.}\ \bibnamefont
  {Rams}}, \bibinfo {author} {\bibfnamefont {P.}~\bibnamefont {Sierant}},
  \bibinfo {author} {\bibfnamefont {O.}~\bibnamefont {Dutta}}, \bibinfo
  {author} {\bibfnamefont {P.}~\bibnamefont {Horodecki}},\ and\ \bibinfo
  {author} {\bibfnamefont {J.}~\bibnamefont {Zakrzewski}},\ }\bibfield  {title}
  {\bibinfo {title} {{At the Limits of Criticality-Based Quantum Metrology:
  Apparent Super-Heisenberg Scaling Revisited}},\ }\href
  {https://doi.org/10.1103/physrevx.8.021022} {\bibfield  {journal} {\bibinfo
  {journal} {Phys. Rev. X}\ }\textbf {\bibinfo {volume} {8}},\ \bibinfo {pages}
  {021022} (\bibinfo {year} {2018})}\BibitemShut {NoStop}%
\bibitem [{\citenamefont {Fr\'erot}\ and\ \citenamefont
  {Roscilde}(2018)}]{Roscilde.2018}%
  \BibitemOpen
  \bibfield  {author} {\bibinfo {author} {\bibfnamefont {I.}~\bibnamefont
  {Fr\'erot}}\ and\ \bibinfo {author} {\bibfnamefont {T.}~\bibnamefont
  {Roscilde}},\ }\bibfield  {title} {\bibinfo {title} {Quantum critical
  metrology},\ }\href {https://doi.org/10.1103/PhysRevLett.121.020402}
  {\bibfield  {journal} {\bibinfo  {journal} {Phys. Rev. Lett.}\ }\textbf
  {\bibinfo {volume} {121}},\ \bibinfo {pages} {020402} (\bibinfo {year}
  {2018})}\BibitemShut {NoStop}%
\bibitem [{\citenamefont {Di~Fresco}\ \emph {et~al.}(2024)\citenamefont
  {Di~Fresco}, \citenamefont {Spagnolo}, \citenamefont {Valenti},\ and\
  \citenamefont {Carollo}}]{DiFresco2024}%
  \BibitemOpen
  \bibfield  {author} {\bibinfo {author} {\bibfnamefont {G.}~\bibnamefont
  {Di~Fresco}}, \bibinfo {author} {\bibfnamefont {B.}~\bibnamefont {Spagnolo}},
  \bibinfo {author} {\bibfnamefont {D.}~\bibnamefont {Valenti}},\ and\ \bibinfo
  {author} {\bibfnamefont {A.}~\bibnamefont {Carollo}},\ }\bibfield  {title}
  {\bibinfo {title} {Metrology and multipartite entanglement in
  measurement-induced phase transition},\ }\href
  {https://doi.org/10.22331/q-2024-04-30-1326} {\bibfield  {journal} {\bibinfo
  {journal} {{Quantum}}\ }\textbf {\bibinfo {volume} {8}},\ \bibinfo {pages}
  {1326} (\bibinfo {year} {2024})}\BibitemShut {NoStop}%
\bibitem [{\citenamefont {Garbe}\ \emph {et~al.}(2019)\citenamefont {Garbe},
  \citenamefont {Bina}, \citenamefont {Keller}, \citenamefont {Paris},\ and\
  \citenamefont {Felicetti}}]{Garbe.2019}%
  \BibitemOpen
  \bibfield  {author} {\bibinfo {author} {\bibfnamefont {L.}~\bibnamefont
  {Garbe}}, \bibinfo {author} {\bibfnamefont {M.}~\bibnamefont {Bina}},
  \bibinfo {author} {\bibfnamefont {A.}~\bibnamefont {Keller}}, \bibinfo
  {author} {\bibfnamefont {M.~G.~A.}\ \bibnamefont {Paris}},\ and\ \bibinfo
  {author} {\bibfnamefont {S.}~\bibnamefont {Felicetti}},\ }\bibfield  {title}
  {\bibinfo {title} {{Critical Quantum Metrology with a Finite-Component
  Quantum Phase Transition.}},\ }\href
  {https://doi.org/10.1103/physrevlett.124.120504} {\bibfield  {journal}
  {\bibinfo  {journal} {Phys. Rev. Lett.}\ }\textbf {\bibinfo {volume} {124}},\
  \bibinfo {pages} {120504} (\bibinfo {year} {2019})},\ \Eprint
  {https://arxiv.org/abs/1910.00604} {1910.00604} \BibitemShut {NoStop}%
\bibitem [{\citenamefont {Chu}\ \emph {et~al.}(2021)\citenamefont {Chu},
  \citenamefont {Zhang}, \citenamefont {Yu},\ and\ \citenamefont
  {Cai}}]{Chu.2021}%
  \BibitemOpen
  \bibfield  {author} {\bibinfo {author} {\bibfnamefont {Y.}~\bibnamefont
  {Chu}}, \bibinfo {author} {\bibfnamefont {S.}~\bibnamefont {Zhang}}, \bibinfo
  {author} {\bibfnamefont {B.}~\bibnamefont {Yu}},\ and\ \bibinfo {author}
  {\bibfnamefont {J.}~\bibnamefont {Cai}},\ }\bibfield  {title} {\bibinfo
  {title} {{Dynamic Framework for Criticality-Enhanced Quantum Sensing}},\
  }\href {https://doi.org/10.1103/physrevlett.126.010502} {\bibfield  {journal}
  {\bibinfo  {journal} {Phys. Rev. Lett.}\ }\textbf {\bibinfo {volume} {126}},\
  \bibinfo {pages} {010502} (\bibinfo {year} {2021})}\BibitemShut {NoStop}%
\bibitem [{\citenamefont {Ilias}\ \emph {et~al.}(2022)\citenamefont {Ilias},
  \citenamefont {Yang}, \citenamefont {Huelga},\ and\ \citenamefont
  {Plenio}}]{Ilias.2022}%
  \BibitemOpen
  \bibfield  {author} {\bibinfo {author} {\bibfnamefont {T.}~\bibnamefont
  {Ilias}}, \bibinfo {author} {\bibfnamefont {D.}~\bibnamefont {Yang}},
  \bibinfo {author} {\bibfnamefont {S.~F.}\ \bibnamefont {Huelga}},\ and\
  \bibinfo {author} {\bibfnamefont {M.~B.}\ \bibnamefont {Plenio}},\ }\bibfield
   {title} {\bibinfo {title} {{Criticality-Enhanced Quantum Sensing via
  Continuous Measurement}},\ }\href
  {https://doi.org/10.1103/prxquantum.3.010354} {\bibfield  {journal} {\bibinfo
   {journal} {PRX Quantum}\ }\textbf {\bibinfo {volume} {3}},\ \bibinfo {pages}
  {010354} (\bibinfo {year} {2022})}\BibitemShut {NoStop}%
\bibitem [{\citenamefont {Garbe}\ \emph
  {et~al.}(2022{\natexlab{a}})\citenamefont {Garbe}, \citenamefont {Abah},
  \citenamefont {Felicetti},\ and\ \citenamefont {Puebla}}]{Garbe_2022}%
  \BibitemOpen
  \bibfield  {author} {\bibinfo {author} {\bibfnamefont {L.}~\bibnamefont
  {Garbe}}, \bibinfo {author} {\bibfnamefont {O.}~\bibnamefont {Abah}},
  \bibinfo {author} {\bibfnamefont {S.}~\bibnamefont {Felicetti}},\ and\
  \bibinfo {author} {\bibfnamefont {R.}~\bibnamefont {Puebla}},\ }\bibfield
  {title} {\bibinfo {title} {Critical quantum metrology with fully-connected
  models: from heisenberg to kibble–zurek scaling},\ }\href
  {https://doi.org/10.1088/2058-9565/ac6ca5} {\bibfield  {journal} {\bibinfo
  {journal} {Quantum Science and Technology}\ }\textbf {\bibinfo {volume}
  {7}},\ \bibinfo {pages} {035010} (\bibinfo {year}
  {2022}{\natexlab{a}})}\BibitemShut {NoStop}%
\bibitem [{\citenamefont {Gietka}\ \emph {et~al.}(2022)\citenamefont {Gietka},
  \citenamefont {Ruks},\ and\ \citenamefont {Busch}}]{Gietka.2022}%
  \BibitemOpen
  \bibfield  {author} {\bibinfo {author} {\bibfnamefont {K.}~\bibnamefont
  {Gietka}}, \bibinfo {author} {\bibfnamefont {L.}~\bibnamefont {Ruks}},\ and\
  \bibinfo {author} {\bibfnamefont {T.}~\bibnamefont {Busch}},\ }\bibfield
  {title} {\bibinfo {title} {{Understanding and Improving Critical Metrology.
  Quenching Superradiant Light-Matter Systems Beyond the Critical Point}},\
  }\href {https://doi.org/10.22331/q-2022-04-27-700} {\bibfield  {journal}
  {\bibinfo  {journal} {Quantum}\ }\textbf {\bibinfo {volume} {6}},\ \bibinfo
  {pages} {700} (\bibinfo {year} {2022})}\BibitemShut {NoStop}%
\bibitem [{\citenamefont {Abah}\ \emph {et~al.}(2022)\citenamefont {Abah},
  \citenamefont {De~Chiara}, \citenamefont {Paternostro},\ and\ \citenamefont
  {Puebla}}]{Abah2022}%
  \BibitemOpen
  \bibfield  {author} {\bibinfo {author} {\bibfnamefont {O.}~\bibnamefont
  {Abah}}, \bibinfo {author} {\bibfnamefont {G.}~\bibnamefont {De~Chiara}},
  \bibinfo {author} {\bibfnamefont {M.}~\bibnamefont {Paternostro}},\ and\
  \bibinfo {author} {\bibfnamefont {R.}~\bibnamefont {Puebla}},\ }\bibfield
  {title} {\bibinfo {title} {Harnessing nonadiabatic excitations promoted by a
  quantum critical point: Quantum battery and spin squeezing},\ }\href
  {https://doi.org/10.1103/PhysRevResearch.4.L022017} {\bibfield  {journal}
  {\bibinfo  {journal} {Phys. Rev. Res.}\ }\textbf {\bibinfo {volume} {4}},\
  \bibinfo {pages} {L022017} (\bibinfo {year} {2022})}\BibitemShut {NoStop}%
\bibitem [{\citenamefont {Garbe}\ \emph
  {et~al.}(2022{\natexlab{b}})\citenamefont {Garbe}, \citenamefont {Abah},
  \citenamefont {Felicetti},\ and\ \citenamefont {Puebla}}]{Garbe2022}%
  \BibitemOpen
  \bibfield  {author} {\bibinfo {author} {\bibfnamefont {L.}~\bibnamefont
  {Garbe}}, \bibinfo {author} {\bibfnamefont {O.}~\bibnamefont {Abah}},
  \bibinfo {author} {\bibfnamefont {S.}~\bibnamefont {Felicetti}},\ and\
  \bibinfo {author} {\bibfnamefont {R.}~\bibnamefont {Puebla}},\ }\bibfield
  {title} {\bibinfo {title} {Exponential time-scaling of estimation precision
  by reaching a quantum critical point},\ }\href
  {https://doi.org/10.1103/PhysRevResearch.4.043061} {\bibfield  {journal}
  {\bibinfo  {journal} {Phys. Rev. Res.}\ }\textbf {\bibinfo {volume} {4}},\
  \bibinfo {pages} {043061} (\bibinfo {year} {2022}{\natexlab{b}})}\BibitemShut
  {NoStop}%
\bibitem [{\citenamefont {Hotter}\ \emph {et~al.}(2024)\citenamefont {Hotter},
  \citenamefont {Ritsch},\ and\ \citenamefont {Gietka}}]{Hotter.2024}%
  \BibitemOpen
  \bibfield  {author} {\bibinfo {author} {\bibfnamefont {C.}~\bibnamefont
  {Hotter}}, \bibinfo {author} {\bibfnamefont {H.}~\bibnamefont {Ritsch}},\
  and\ \bibinfo {author} {\bibfnamefont {K.}~\bibnamefont {Gietka}},\
  }\bibfield  {title} {\bibinfo {title} {Combining critical and quantum
  metrology},\ }\href {https://doi.org/10.1103/PhysRevLett.132.060801}
  {\bibfield  {journal} {\bibinfo  {journal} {Phys. Rev. Lett.}\ }\textbf
  {\bibinfo {volume} {132}},\ \bibinfo {pages} {060801} (\bibinfo {year}
  {2024})}\BibitemShut {NoStop}%
\bibitem [{\citenamefont {Liu}\ \emph {et~al.}(2021)\citenamefont {Liu},
  \citenamefont {Chen}, \citenamefont {Jiang}, \citenamefont {Yang},
  \citenamefont {Wu}, \citenamefont {Li}, \citenamefont {Yuan}, \citenamefont
  {Peng},\ and\ \citenamefont {Du}}]{Liu.2021}%
  \BibitemOpen
  \bibfield  {author} {\bibinfo {author} {\bibfnamefont {R.}~\bibnamefont
  {Liu}}, \bibinfo {author} {\bibfnamefont {Y.}~\bibnamefont {Chen}}, \bibinfo
  {author} {\bibfnamefont {M.}~\bibnamefont {Jiang}}, \bibinfo {author}
  {\bibfnamefont {X.}~\bibnamefont {Yang}}, \bibinfo {author} {\bibfnamefont
  {Z.}~\bibnamefont {Wu}}, \bibinfo {author} {\bibfnamefont {Y.}~\bibnamefont
  {Li}}, \bibinfo {author} {\bibfnamefont {H.}~\bibnamefont {Yuan}}, \bibinfo
  {author} {\bibfnamefont {X.}~\bibnamefont {Peng}},\ and\ \bibinfo {author}
  {\bibfnamefont {J.}~\bibnamefont {Du}},\ }\bibfield  {title} {\bibinfo
  {title} {{Experimental critical quantum metrology with the Heisenberg
  scaling}},\ }\href {https://doi.org/10.1038/s41534-021-00507-x} {\bibfield
  {journal} {\bibinfo  {journal} {npj Quantum Inf.}\ }\textbf {\bibinfo
  {volume} {7}},\ \bibinfo {pages} {170} (\bibinfo {year} {2021})}\BibitemShut
  {NoStop}%
\bibitem [{\citenamefont {Ding}\ \emph {et~al.}(2022)\citenamefont {Ding},
  \citenamefont {Liu}, \citenamefont {Shi}, \citenamefont {Guo}, \citenamefont
  {Mølmer},\ and\ \citenamefont {Adams}}]{Ding.2022}%
  \BibitemOpen
  \bibfield  {author} {\bibinfo {author} {\bibfnamefont {D.-S.}\ \bibnamefont
  {Ding}}, \bibinfo {author} {\bibfnamefont {Z.-K.}\ \bibnamefont {Liu}},
  \bibinfo {author} {\bibfnamefont {B.-S.}\ \bibnamefont {Shi}}, \bibinfo
  {author} {\bibfnamefont {G.-C.}\ \bibnamefont {Guo}}, \bibinfo {author}
  {\bibfnamefont {K.}~\bibnamefont {Mølmer}},\ and\ \bibinfo {author}
  {\bibfnamefont {C.~S.}\ \bibnamefont {Adams}},\ }\bibfield  {title} {\bibinfo
  {title} {{Enhanced metrology at the critical point of a many-body Rydberg
  atomic system}},\ }\href {https://doi.org/10.1038/s41567-022-01777-8}
  {\bibfield  {journal} {\bibinfo  {journal} {Nat. Phys.}\ }\textbf {\bibinfo
  {volume} {18}},\ \bibinfo {pages} {1447} (\bibinfo {year}
  {2022})}\BibitemShut {NoStop}%
\bibitem [{\citenamefont {Beaulieu}\ \emph {et~al.}(2025)\citenamefont
  {Beaulieu}, \citenamefont {Minganti}, \citenamefont {Frasca}, \citenamefont
  {Scigliuzzo}, \citenamefont {Felicetti}, \citenamefont {Di~Candia},\ and\
  \citenamefont {Scarlino}}]{Beaulieu.2025}%
  \BibitemOpen
  \bibfield  {author} {\bibinfo {author} {\bibfnamefont {G.}~\bibnamefont
  {Beaulieu}}, \bibinfo {author} {\bibfnamefont {F.}~\bibnamefont {Minganti}},
  \bibinfo {author} {\bibfnamefont {S.}~\bibnamefont {Frasca}}, \bibinfo
  {author} {\bibfnamefont {M.}~\bibnamefont {Scigliuzzo}}, \bibinfo {author}
  {\bibfnamefont {S.}~\bibnamefont {Felicetti}}, \bibinfo {author}
  {\bibfnamefont {R.}~\bibnamefont {Di~Candia}},\ and\ \bibinfo {author}
  {\bibfnamefont {P.}~\bibnamefont {Scarlino}},\ }\bibfield  {title} {\bibinfo
  {title} {Criticality-enhanced quantum sensing with a parametric
  superconducting resonator},\ }\href
  {https://doi.org/10.1103/PRXQuantum.6.020301} {\bibfield  {journal} {\bibinfo
   {journal} {PRX Quantum}\ }\textbf {\bibinfo {volume} {6}},\ \bibinfo {pages}
  {020301} (\bibinfo {year} {2025})}\BibitemShut {NoStop}%
\bibitem [{\citenamefont {Ribeiro}\ \emph {et~al.}(2007)\citenamefont
  {Ribeiro}, \citenamefont {Vidal},\ and\ \citenamefont
  {Mosseri}}]{Ribeiro.2007}%
  \BibitemOpen
  \bibfield  {author} {\bibinfo {author} {\bibfnamefont {P.}~\bibnamefont
  {Ribeiro}}, \bibinfo {author} {\bibfnamefont {J.}~\bibnamefont {Vidal}},\
  and\ \bibinfo {author} {\bibfnamefont {R.}~\bibnamefont {Mosseri}},\
  }\bibfield  {title} {\bibinfo {title} {{Thermodynamical Limit of the
  Lipkin-Meshkov-Glick Model}},\ }\href
  {https://doi.org/10.1103/physrevlett.99.050402} {\bibfield  {journal}
  {\bibinfo  {journal} {Phys. Rev. Lett.}\ }\textbf {\bibinfo {volume} {99}},\
  \bibinfo {pages} {43 } (\bibinfo {year} {2007})},\ \Eprint
  {https://arxiv.org/abs/cond-mat/0703490} {cond-mat/0703490} \BibitemShut
  {NoStop}%
\bibitem [{\citenamefont {Emary}\ and\ \citenamefont
  {Brandes}(2003)}]{Emary.2003}%
  \BibitemOpen
  \bibfield  {author} {\bibinfo {author} {\bibfnamefont {C.}~\bibnamefont
  {Emary}}\ and\ \bibinfo {author} {\bibfnamefont {T.}~\bibnamefont
  {Brandes}},\ }\bibfield  {title} {\bibinfo {title} {{Quantum Chaos Triggered
  by Precursors of a Quantum Phase Transition: The Dicke Model}},\ }\href@noop
  {} {\bibfield  {journal} {\bibinfo  {journal} {Phys. Rev. Lett.}\ }\textbf
  {\bibinfo {volume} {90}},\ \bibinfo {pages} {42 } (\bibinfo {year}
  {2003})}\BibitemShut {NoStop}%
\bibitem [{\citenamefont {Hwang}\ \emph {et~al.}(2015)\citenamefont {Hwang},
  \citenamefont {Puebla},\ and\ \citenamefont {Plenio}}]{Hwang.2015}%
  \BibitemOpen
  \bibfield  {author} {\bibinfo {author} {\bibfnamefont {M.-J.}\ \bibnamefont
  {Hwang}}, \bibinfo {author} {\bibfnamefont {R.}~\bibnamefont {Puebla}},\ and\
  \bibinfo {author} {\bibfnamefont {M.~B.}\ \bibnamefont {Plenio}},\ }\bibfield
   {title} {\bibinfo {title} {{Quantum Phase Transition and Universal Dynamics
  in the Rabi Model}},\ }\href@noop {} {\bibfield  {journal} {\bibinfo
  {journal} {Phys. Rev. Lett.}\ }\textbf {\bibinfo {volume} {115}},\ \bibinfo
  {pages} {180404} (\bibinfo {year} {2015})}\BibitemShut {NoStop}%
\bibitem [{\citenamefont {Felicetti}\ and\ \citenamefont
  {Boité}(2020)}]{Felicetti.2020}%
  \BibitemOpen
  \bibfield  {author} {\bibinfo {author} {\bibfnamefont {S.}~\bibnamefont
  {Felicetti}}\ and\ \bibinfo {author} {\bibfnamefont {A.~L.}\ \bibnamefont
  {Boité}},\ }\bibfield  {title} {\bibinfo {title} {{Universal Spectral
  Features of Ultrastrongly Coupled Systems}},\ }\href
  {https://doi.org/10.1103/physrevlett.124.040404} {\bibfield  {journal}
  {\bibinfo  {journal} {Phys. Rev. Lett.}\ }\textbf {\bibinfo {volume} {124}},\
  \bibinfo {pages} {040404} (\bibinfo {year} {2020})}\BibitemShut {NoStop}%
\bibitem [{\citenamefont {Cai}\ \emph {et~al.}(2021)\citenamefont {Cai},
  \citenamefont {Liu}, \citenamefont {Zhao}, \citenamefont {Wu}, \citenamefont
  {Mei}, \citenamefont {Jiang}, \citenamefont {He}, \citenamefont {Zhang},
  \citenamefont {Zhou},\ and\ \citenamefont {Duan}}]{Cai.2021}%
  \BibitemOpen
  \bibfield  {author} {\bibinfo {author} {\bibfnamefont {M.-L.}\ \bibnamefont
  {Cai}}, \bibinfo {author} {\bibfnamefont {Z.-D.}\ \bibnamefont {Liu}},
  \bibinfo {author} {\bibfnamefont {W.-D.}\ \bibnamefont {Zhao}}, \bibinfo
  {author} {\bibfnamefont {Y.-K.}\ \bibnamefont {Wu}}, \bibinfo {author}
  {\bibfnamefont {Q.-X.}\ \bibnamefont {Mei}}, \bibinfo {author} {\bibfnamefont
  {Y.}~\bibnamefont {Jiang}}, \bibinfo {author} {\bibfnamefont
  {L.}~\bibnamefont {He}}, \bibinfo {author} {\bibfnamefont {X.}~\bibnamefont
  {Zhang}}, \bibinfo {author} {\bibfnamefont {Z.-C.}\ \bibnamefont {Zhou}},\
  and\ \bibinfo {author} {\bibfnamefont {L.-M.}\ \bibnamefont {Duan}},\
  }\bibfield  {title} {\bibinfo {title} {{Observation of a quantum phase
  transition in the quantum Rabi model with a single trapped ion}},\ }\href
  {https://doi.org/10.1038/s41467-021-21425-8} {\bibfield  {journal} {\bibinfo
  {journal} {Nat. Commun.}\ }\textbf {\bibinfo {volume} {12}},\ \bibinfo
  {pages} {1126} (\bibinfo {year} {2021})}\BibitemShut {NoStop}%
\bibitem [{\citenamefont {Chen}\ \emph {et~al.}(2021)\citenamefont {Chen},
  \citenamefont {Wu}, \citenamefont {Jiang}, \citenamefont {Lü}, \citenamefont
  {Peng},\ and\ \citenamefont {Du}}]{Chen.2021}%
  \BibitemOpen
  \bibfield  {author} {\bibinfo {author} {\bibfnamefont {X.}~\bibnamefont
  {Chen}}, \bibinfo {author} {\bibfnamefont {Z.}~\bibnamefont {Wu}}, \bibinfo
  {author} {\bibfnamefont {M.}~\bibnamefont {Jiang}}, \bibinfo {author}
  {\bibfnamefont {X.-Y.}\ \bibnamefont {Lü}}, \bibinfo {author} {\bibfnamefont
  {X.}~\bibnamefont {Peng}},\ and\ \bibinfo {author} {\bibfnamefont
  {J.}~\bibnamefont {Du}},\ }\bibfield  {title} {\bibinfo {title}
  {{Experimental quantum simulation of superradiant phase transition beyond
  no-go theorem via antisqueezing}},\ }\href
  {https://doi.org/10.1038/s41467-021-26573-5} {\bibfield  {journal} {\bibinfo
  {journal} {Nat. Commun.}\ }\textbf {\bibinfo {volume} {12}},\ \bibinfo
  {pages} {6281} (\bibinfo {year} {2021})}\BibitemShut {NoStop}%
\bibitem [{\citenamefont {Zheng}\ \emph {et~al.}(2023)\citenamefont {Zheng},
  \citenamefont {Ning}, \citenamefont {Chen}, \citenamefont {Lü},
  \citenamefont {Shen}, \citenamefont {Xu}, \citenamefont {Zhang},
  \citenamefont {Xu}, \citenamefont {Li}, \citenamefont {Xia}, \citenamefont
  {Wu}, \citenamefont {Yang}, \citenamefont {Miranowicz}, \citenamefont
  {Lambert}, \citenamefont {Zheng}, \citenamefont {Fan}, \citenamefont {Nori},\
  and\ \citenamefont {Zheng}}]{Zheng.2023}%
  \BibitemOpen
  \bibfield  {author} {\bibinfo {author} {\bibfnamefont {R.-H.}\ \bibnamefont
  {Zheng}}, \bibinfo {author} {\bibfnamefont {W.}~\bibnamefont {Ning}},
  \bibinfo {author} {\bibfnamefont {Y.-H.}\ \bibnamefont {Chen}}, \bibinfo
  {author} {\bibfnamefont {J.-H.}\ \bibnamefont {Lü}}, \bibinfo {author}
  {\bibfnamefont {L.-T.}\ \bibnamefont {Shen}}, \bibinfo {author}
  {\bibfnamefont {K.}~\bibnamefont {Xu}}, \bibinfo {author} {\bibfnamefont
  {Y.-R.}\ \bibnamefont {Zhang}}, \bibinfo {author} {\bibfnamefont
  {D.}~\bibnamefont {Xu}}, \bibinfo {author} {\bibfnamefont {H.}~\bibnamefont
  {Li}}, \bibinfo {author} {\bibfnamefont {Y.}~\bibnamefont {Xia}}, \bibinfo
  {author} {\bibfnamefont {F.}~\bibnamefont {Wu}}, \bibinfo {author}
  {\bibfnamefont {Z.-B.}\ \bibnamefont {Yang}}, \bibinfo {author}
  {\bibfnamefont {A.}~\bibnamefont {Miranowicz}}, \bibinfo {author}
  {\bibfnamefont {N.}~\bibnamefont {Lambert}}, \bibinfo {author} {\bibfnamefont
  {D.}~\bibnamefont {Zheng}}, \bibinfo {author} {\bibfnamefont
  {H.}~\bibnamefont {Fan}}, \bibinfo {author} {\bibfnamefont {F.}~\bibnamefont
  {Nori}},\ and\ \bibinfo {author} {\bibfnamefont {S.-B.}\ \bibnamefont
  {Zheng}},\ }\bibfield  {title} {\bibinfo {title} {{Observation of a
  Superradiant Phase Transition with Emergent Cat States}},\ }\href
  {https://doi.org/10.1103/physrevlett.131.113601} {\bibfield  {journal}
  {\bibinfo  {journal} {Phys. Rev. Lett.}\ }\textbf {\bibinfo {volume} {131}},\
  \bibinfo {pages} {113601} (\bibinfo {year} {2023})}\BibitemShut {NoStop}%
\bibitem [{\citenamefont {Wu}\ \emph {et~al.}(2024)\citenamefont {Wu},
  \citenamefont {Hu}, \citenamefont {Wang}, \citenamefont {Chen}, \citenamefont
  {Li}, \citenamefont {Zhao}, \citenamefont {Lü},\ and\ \citenamefont
  {Peng}}]{Wu.2024}%
  \BibitemOpen
  \bibfield  {author} {\bibinfo {author} {\bibfnamefont {Z.}~\bibnamefont
  {Wu}}, \bibinfo {author} {\bibfnamefont {C.}~\bibnamefont {Hu}}, \bibinfo
  {author} {\bibfnamefont {T.}~\bibnamefont {Wang}}, \bibinfo {author}
  {\bibfnamefont {Y.}~\bibnamefont {Chen}}, \bibinfo {author} {\bibfnamefont
  {Y.}~\bibnamefont {Li}}, \bibinfo {author} {\bibfnamefont {L.}~\bibnamefont
  {Zhao}}, \bibinfo {author} {\bibfnamefont {X.-Y.}\ \bibnamefont {Lü}},\ and\
  \bibinfo {author} {\bibfnamefont {X.}~\bibnamefont {Peng}},\ }\bibfield
  {title} {\bibinfo {title} {{Experimental Quantum Simulation of
  Multicriticality in Closed and Open Rabi Model}},\ }\href
  {https://doi.org/10.1103/physrevlett.133.173602} {\bibfield  {journal}
  {\bibinfo  {journal} {Phys. Rev. Lett.}\ }\textbf {\bibinfo {volume} {133}},\
  \bibinfo {pages} {173602} (\bibinfo {year} {2024})}\BibitemShut {NoStop}%
\bibitem [{\citenamefont {Ilias}\ \emph {et~al.}(2024)\citenamefont {Ilias},
  \citenamefont {Yang}, \citenamefont {Huelga},\ and\ \citenamefont
  {Plenio}}]{Ilias.2024}%
  \BibitemOpen
  \bibfield  {author} {\bibinfo {author} {\bibfnamefont {T.}~\bibnamefont
  {Ilias}}, \bibinfo {author} {\bibfnamefont {D.}~\bibnamefont {Yang}},
  \bibinfo {author} {\bibfnamefont {S.~F.}\ \bibnamefont {Huelga}},\ and\
  \bibinfo {author} {\bibfnamefont {M.~B.}\ \bibnamefont {Plenio}},\ }\bibfield
   {title} {\bibinfo {title} {{Criticality-enhanced electric field gradient
  sensor with single trapped ions}},\ }\href
  {https://doi.org/10.1038/s41534-024-00833-w} {\bibfield  {journal} {\bibinfo
  {journal} {npj Quantum Inf.}\ }\textbf {\bibinfo {volume} {10}},\ \bibinfo
  {pages} {36} (\bibinfo {year} {2024})}\BibitemShut {NoStop}%
\bibitem [{\citenamefont {Yuan}\ and\ \citenamefont {Fung}(2015)}]{Yuan15}%
  \BibitemOpen
  \bibfield  {author} {\bibinfo {author} {\bibfnamefont {H.}~\bibnamefont
  {Yuan}}\ and\ \bibinfo {author} {\bibfnamefont {C.-H.~F.}\ \bibnamefont
  {Fung}},\ }\bibfield  {title} {\bibinfo {title} {Optimal feedback scheme and
  universal time scaling for hamiltonian parameter estimation},\ }\href
  {https://doi.org/10.1103/PhysRevLett.115.110401} {\bibfield  {journal}
  {\bibinfo  {journal} {Phys. Rev. Lett.}\ }\textbf {\bibinfo {volume} {115}},\
  \bibinfo {pages} {110401} (\bibinfo {year} {2015})}\BibitemShut {NoStop}%
\bibitem [{\citenamefont {Pang}\ and\ \citenamefont {Jordan}(2017)}]{Pang17}%
  \BibitemOpen
  \bibfield  {author} {\bibinfo {author} {\bibfnamefont {S.}~\bibnamefont
  {Pang}}\ and\ \bibinfo {author} {\bibfnamefont {A.~N.}\ \bibnamefont
  {Jordan}},\ }\bibfield  {title} {\bibinfo {title} {Optimal adaptive control
  for quantum metrology with time-dependent hamiltonians},\ }\href@noop {}
  {\bibfield  {journal} {\bibinfo  {journal} {Nat. Commun.}\ }\textbf {\bibinfo
  {volume} {8}},\ \bibinfo {pages} {14695} (\bibinfo {year}
  {2017})}\BibitemShut {NoStop}%
\bibitem [{\citenamefont {Yang}\ \emph {et~al.}(2022)\citenamefont {Yang},
  \citenamefont {Pang}, \citenamefont {Chen}, \citenamefont {Jordan},\ and\
  \citenamefont {del Campo}}]{Yang22}%
  \BibitemOpen
  \bibfield  {author} {\bibinfo {author} {\bibfnamefont {J.}~\bibnamefont
  {Yang}}, \bibinfo {author} {\bibfnamefont {S.}~\bibnamefont {Pang}}, \bibinfo
  {author} {\bibfnamefont {Z.}~\bibnamefont {Chen}}, \bibinfo {author}
  {\bibfnamefont {A.~N.}\ \bibnamefont {Jordan}},\ and\ \bibinfo {author}
  {\bibfnamefont {A.}~\bibnamefont {del Campo}},\ }\bibfield  {title} {\bibinfo
  {title} {Variational principle for optimal quantum controls in quantum
  metrology},\ }\href {https://doi.org/10.1103/PhysRevLett.128.160505}
  {\bibfield  {journal} {\bibinfo  {journal} {Phys. Rev. Lett.}\ }\textbf
  {\bibinfo {volume} {128}},\ \bibinfo {pages} {160505} (\bibinfo {year}
  {2022})}\BibitemShut {NoStop}%
\bibitem [{\citenamefont {Bakemeier}\ \emph {et~al.}(2012)\citenamefont
  {Bakemeier}, \citenamefont {Alvermann},\ and\ \citenamefont
  {Fehske}}]{Bakemeier.2012}%
  \BibitemOpen
  \bibfield  {author} {\bibinfo {author} {\bibfnamefont {L.}~\bibnamefont
  {Bakemeier}}, \bibinfo {author} {\bibfnamefont {A.}~\bibnamefont
  {Alvermann}},\ and\ \bibinfo {author} {\bibfnamefont {H.}~\bibnamefont
  {Fehske}},\ }\bibfield  {title} {\bibinfo {title} {{Quantum phase transition
  in the Dicke model with critical and noncritical entanglement}},\ }\href
  {https://doi.org/10.1103/physreva.85.043821} {\bibfield  {journal} {\bibinfo
  {journal} {Phys. Rev. A}\ }\textbf {\bibinfo {volume} {85}},\ \bibinfo
  {pages} {043821} (\bibinfo {year} {2012})}\BibitemShut {NoStop}%
\bibitem [{\citenamefont {Weedbrook}\ \emph {et~al.}(2012)\citenamefont
  {Weedbrook}, \citenamefont {Pirandola}, \citenamefont {Garc\'{\i}a-Patr\'on},
  \citenamefont {Cerf}, \citenamefont {Ralph}, \citenamefont {Shapiro},\ and\
  \citenamefont {Lloyd}}]{RevModPhys.84.621}%
  \BibitemOpen
  \bibfield  {author} {\bibinfo {author} {\bibfnamefont {C.}~\bibnamefont
  {Weedbrook}}, \bibinfo {author} {\bibfnamefont {S.}~\bibnamefont
  {Pirandola}}, \bibinfo {author} {\bibfnamefont {R.}~\bibnamefont
  {Garc\'{\i}a-Patr\'on}}, \bibinfo {author} {\bibfnamefont {N.~J.}\
  \bibnamefont {Cerf}}, \bibinfo {author} {\bibfnamefont {T.~C.}\ \bibnamefont
  {Ralph}}, \bibinfo {author} {\bibfnamefont {J.~H.}\ \bibnamefont {Shapiro}},\
  and\ \bibinfo {author} {\bibfnamefont {S.}~\bibnamefont {Lloyd}},\ }\bibfield
   {title} {\bibinfo {title} {Gaussian quantum information},\ }\href
  {https://doi.org/10.1103/RevModPhys.84.621} {\bibfield  {journal} {\bibinfo
  {journal} {Rev. Mod. Phys.}\ }\textbf {\bibinfo {volume} {84}},\ \bibinfo
  {pages} {621} (\bibinfo {year} {2012})}\BibitemShut {NoStop}%
\bibitem [{\citenamefont {Downing}\ and\ \citenamefont
  {Ukhtary}(2024)}]{Downing2024}%
  \BibitemOpen
  \bibfield  {author} {\bibinfo {author} {\bibfnamefont {C.~A.}\ \bibnamefont
  {Downing}}\ and\ \bibinfo {author} {\bibfnamefont {M.~S.}\ \bibnamefont
  {Ukhtary}},\ }\bibfield  {title} {\bibinfo {title} {Hyperbolic enhancement of
  a quantum battery},\ }\href {https://doi.org/10.1103/PhysRevA.109.052206}
  {\bibfield  {journal} {\bibinfo  {journal} {Phys. Rev. A}\ }\textbf {\bibinfo
  {volume} {109}},\ \bibinfo {pages} {052206} (\bibinfo {year}
  {2024})}\BibitemShut {NoStop}%
\bibitem [{\citenamefont {Pinel}\ \emph {et~al.}(2013)\citenamefont {Pinel},
  \citenamefont {Jian}, \citenamefont {Treps}, \citenamefont {Fabre},\ and\
  \citenamefont {Braun}}]{Pinel2013}%
  \BibitemOpen
  \bibfield  {author} {\bibinfo {author} {\bibfnamefont {O.}~\bibnamefont
  {Pinel}}, \bibinfo {author} {\bibfnamefont {P.}~\bibnamefont {Jian}},
  \bibinfo {author} {\bibfnamefont {N.}~\bibnamefont {Treps}}, \bibinfo
  {author} {\bibfnamefont {C.}~\bibnamefont {Fabre}},\ and\ \bibinfo {author}
  {\bibfnamefont {D.}~\bibnamefont {Braun}},\ }\bibfield  {title} {\bibinfo
  {title} {Quantum parameter estimation using general single-mode gaussian
  states},\ }\href {https://doi.org/10.1103/PhysRevA.88.040102} {\bibfield
  {journal} {\bibinfo  {journal} {Phys. Rev. A}\ }\textbf {\bibinfo {volume}
  {88}},\ \bibinfo {pages} {040102} (\bibinfo {year} {2013})}\BibitemShut
  {NoStop}%
\end{thebibliography}%

\appendix
\widetext
\section{Dynamics of squeezed states and the quantum Fisher information}\label{section:dynamic}

\subsection{Dynamics of fully connected system}
In this section, we derive the equations of motion under the following effective time-dependent Hamiltonian of the fully connected system:
\begin{equation}
\label{eq:effective_H}
\h(t) = \omega \an^\dagger\an - \frac{\epsilon(t)}{4} (\an^\dagger+\an)^2.
\end{equation}
Let us express the quantum state at time $t$ as $\ket{\psi_t} = \U_t \ket{\psi_0}$, where the dynamics of $U_t$ is given by the Schr\"odinger equation:
\begin{equation}
\label{eq:schrodinger}
\frac{d \U_t}{dt} = -i \hat{H}(t) \U_t.
\end{equation}
By noting that Eq.~\eqref{eq:effective_H} only contains quadratic terms in $\hat{a}$ and $\hat{a}^\dagger$, the evolution unitary $U_t$ can be fully characterized by the following Gaussian unitary operations without displacement:
\begin{equation}
\label{eq:unitary_form}
	\U_t = e^{i \chi} \hat{R}\left(\frac{\varphi}{2}\right) \hat{S}\left( r \right)\hat{R}\left(\frac{\theta}{2}\right),
\end{equation}
where $\hat R(\phi) = e^{-i \phi \an^\dagger\an}$ is the rotation operator and $\hat{S}(r) = e^{\frac{r}{2}(\an^2 - \an^{\dagger 2})}$ is the squeezing operator with a real-valued squeezing $r$. We note that when $\hat{U}_t$ is acting on the vacuum state $\ket{\psi_0} = \ket{0}$, the resulting state becomes a squeezed vacuum state, 
\begin{equation}
\ket{\psi_t} = \hat{U}_t \ket{0} = \hat{S}(\xi) \ket{0},
\end{equation}
where $\hat{S}(\xi) = e^{\frac{1}{2}(\xi^* \an^2 - \xi \an^{\dagger 2})}$ with $\xi = r e^{-i\varphi}$. We note that the state is independent of $\theta$ and the overall phase $\chi$ does not have any physical significance.

From Eq.~\eqref{eq:schrodinger}, we obtain
\begin{equation}
\begin{aligned}
\frac{d \U_t}{dt} &= 
\frac{d e^{i \chi}}{dt} \hat{R}\left(\frac{\varphi}{2}\right) \hat{S}\left( r \right)\hat{R}\left(\frac{\theta}{2}\right) + e^{i\chi} \left[
\left(\frac{d \hat{R}\left(\frac{\varphi}{2}\right)}{dt} \right) \hat{S}\left( r \right)\hat{R}\left(\frac{\theta}{2}\right) + \hat{R}\left(\frac{\varphi}{2}\right) \left( \frac{d \hat{S}\left( r \right)}{dt} \right) \hat{R}\left(\frac{\theta}{2}\right) + \hat{R}\left(\frac{\varphi}{2}\right) \hat{S}\left( r \right) \left( \frac{d \hat{R}\left(\frac{\theta}{2}\right)}{dt} \right) \right] \\
&=-\frac{i}{2} e^{i\chi} \left[  -2\dot \chi + \dot\varphi (\hat{a}^\dagger \hat{a}) \hat{R}\left(\frac{\varphi}{2}\right) \hat{S}\left( r \right)\hat{R}\left(\frac{\theta}{2}\right) + i \dot r \hat{R}\left(\frac{\varphi}{2}\right) (\hat{a}^2 - \hat{a}^{\dagger 2}) \hat{S}\left( r \right)\hat{R}\left(\frac{\theta}{2}\right) + \dot \theta  \hat{R}\left(\frac{\varphi}{2}\right) \hat{S}\left( r \right)\hat{R}\left(\frac{\theta}{2}\right) (\hat{a}^\dagger\hat{a}) \right]\\
&=-\frac{i}{2} e^{i\chi} \left[  -2\dot \chi + \dot\varphi (\hat{a}^\dagger \hat{a})  + i \dot r \hat{R}\left(\frac{\varphi}{2}\right) (\hat{a}^2 - \hat{a}^{\dagger 2}) \hat{R}^\dagger \left(\frac{\varphi}{2}\right) + \dot \theta  \hat{R}\left(\frac{\varphi}{2}\right) \hat{S}\left( r \right) (\hat{a}^\dagger\hat{a}) \hat{S}^\dagger \left( r \right) \hat{R}^\dagger\left(\frac{\varphi}{2}\right) \right] \hat{R}\left(\frac{\varphi}{2}\right) \hat{S}\left( r \right)\hat{R}\left(\frac{\theta}{2}\right)\\
&=-\frac{i}{2} e^{i\chi} \left[  -2\dot \chi + \dot\varphi (\hat{a}^\dagger \hat{a})  +  i\dot r \hat{R}\left(\frac{\varphi}{2}\right) (\hat{a}^2 - \hat{a}^{\dagger 2}) \hat{R}^\dagger \left(\frac{\varphi}{2}\right) + \dot \theta  \hat{R}\left(\frac{\varphi}{2}\right) \hat{S}\left( r \right) (\hat{a}\hat{a}^\dagger) \hat{S}^\dagger \left( r \right) \hat{R}^\dagger\left(\frac{\varphi}{2}\right) \right] \U_t \\
&= (-i) \hat{H}(t) \U_t.
\end{aligned}
\end{equation}
This leads to the following expression:
\begin{equation}
\label{eq:one_line_eq}
\frac{1}{2} \left[  -2\dot \chi + \dot\varphi (\hat{a}^\dagger\hat{a})  + i \dot r \hat{R}\left(\frac{\varphi}{2}\right) (\hat{a}^2 - \hat{a}^{\dagger 2}) \hat{R}^\dagger \left(\frac{\varphi}{2}\right) + \dot \theta  \hat{R}\left(\frac{\varphi}{2}\right) \hat{S}\left( r \right) (\hat{a}^\dagger\hat{a}) \hat{S}^\dagger \left( r \right) \hat{R}^\dagger\left(\frac{\varphi}{2}\right) \right]  = \hat{H}(t) = \omega \an^\dagger\an - \frac{\epsilon(t)}{4} (\an^\dagger+\an)^2.
\end{equation}
From the action of the rotation and squeezing operations to the bosonic operators,
\begin{eqnarray}
    \hat{R}\left(\frac{\varphi}{2}\right)\an\hat{R}^\dagger\left(\frac{\varphi}{2}\right)&=&e^{i\frac{\varphi}{2}}\an,\\
    \hat{S}\left(r \right)\an \hat{S}^\dagger\left(r\right)&=&\cosh\left(r\right)\an + \sinh\left(r\right)\an^\dagger,
\end{eqnarray}
Eq.~\eqref{eq:one_line_eq} can be simplified as
\begin{equation}
\frac{1}{2} \left[  -2\dot \chi + \dot\varphi (\hat{a}^\dagger\hat{a})  + i \dot r (e^{i\varphi} \hat{a}^2 - e^{-i\varphi} \hat{a}^{\dagger 2})  + \dot \theta  \left( \cosh(2r) \hat{a}^\dagger \hat{a} + \frac{e^{i\varphi} \hat{a}^2 + e^{-i\varphi} \hat{a}^{\dagger 2}} {2}\sinh(2r)+ \sinh^2(r) \mathbb{I} \right) \right] = \omega \an^\dagger\an - \frac{\epsilon(t)}{4} (\an^\dagger+\an)^2.
\end{equation}

By collecting each order of $\mathbb{I}$, $\hat{a}^\dagger \hat{a}$, $\hat{a}^2$, and $\hat{a}^{\dagger 2}$, we obtain
\begin{align}
- 2\dot \chi + \dot \theta \sinh^2(r) &= - \frac{\epsilon(t)}{2}, \\
\dot \varphi + \dot \theta \cosh(2r) &= 2\omega - \epsilon(t), \\
\label{eq:d1}
e^{i\varphi} \left[ i \dot r + \frac{\dot \theta}{2} \sinh(2r) \right] &= -\frac{\epsilon(t)}{2}, \\
\label{eq:d2}
e^{-i\varphi} \left[ -i \dot r + \frac{\dot \theta}{2} \sinh(2r) \right] &= -\frac{\epsilon(t)}{2} .
\end{align}
From Eqs.~\eqref{eq:d1} and \eqref{eq:d2}, we obtain
\begin{equation}
\label{eq:Eq_theta}
	\dot \theta = -\frac{\epsilon(t) \cos\varphi}{\sinh(2r)}.
\end{equation}
Finally, substituting the above expressions and rearranging the equations leads to the desired equation of motion,
\begin{equation}
\label{eq:Eq_M}
\begin{aligned}
	\dot r  &= \frac{\epsilon(t)}{2} \sin\varphi, \\
	\dot \varphi  &= 2\omega - \epsilon(t) (1 - \coth(2r) \cos\varphi),\\
	\dot \chi  &= \frac{\epsilon(t)}{4} \left( 1 - \frac{ 2 \sinh^2(r) \cos\varphi}{\sinh(2r)} \right),
\end{aligned}
\end{equation}
where we note again that the overall phase $\chi$ does not have any physical significance.

\subsection{Quantum Fisher information for dynamical encoding}
We provide the closed form of the quantum Fisher information (QFI) when estimating the system parameter $\omega = \omega_0 + \delta \omega$ encoded in the evolved state after time $T$, $\ket{\psi_T} = \hat{U}_T \ket{\psi_T}$. The explicit form of the QFI is given as
\begin{equation}
\mathcal{F}_\omega = 4(\mean{\partial_\omega \psi_T|\partial_\omega \psi_T}-|\mean{ \psi_T|\partial_\omega \psi_T}|^2)\at[]{\omega =\omega_0},
\end{equation}
where $\ket{\partial_\omega \psi_T} = \frac{\partial \ket{\psi_T}}{\partial \omega}  = \frac{\partial \hat{U}_T}{\partial \omega}  \ket{\psi_0}$.
We note that by introducing the following hermitian operator,
\begin{equation}
\hat{G}_T =  i \hat{U}_T^\dagger \left( \frac{\partial \hat{U}_T}{ \partial \omega} \right) \at[]{\omega =\omega_0},
\end{equation}
the QFI can be expressed as the expectation value over the initial state $\ket{\psi_0}$ as
\begin{equation}
\begin{aligned}
\mathcal{F}_\omega 
&= 4\left( \bra{\psi_0} \left( \frac{\partial \hat{U}_T^\dagger}{ \partial \omega} \right)   \left( \frac{\partial \hat{U}_T}{ \partial \omega} \right)  \ket{\psi_0} - \left|\bra{\psi_0} i \hat{U}_T^\dagger \left( \frac{\partial \hat{U}_T}{ \partial \omega} \right) \ket{\psi_0}\right|^2 \right) \at[]{\omega =\omega_0}\\
&= 4\left( \bra{\psi_0} \left[(-i) \left( \frac{\partial \hat{U}_T^\dagger}{ \partial \omega} \right) \hat{U}_T \right] \left[ i \hat{U}_T^\dagger \left( \frac{\partial \hat{U}_T}{ \partial \omega} \right) \right] \ket{\psi_0} - \left|\bra{\psi_0} i \hat{U}_T^\dagger \left( \frac{\partial \hat{U}_T}{ \partial \omega} \right) \ket{\psi_0}\right|^2 \right) \at[]{\omega =\omega_0}\\
&= 4( \langle \psi_0 | \hat{G}_T^2 | \psi_0 \rangle -  \langle \psi_0 | \hat{G}_T | \psi_0 \rangle^2  \\
&= 4 {\rm Var}_{\ket{\psi_0}} (\hat{G}_T),
\end{aligned}
\end{equation}
where $\hat{G}_T = \hat{G}_T^\dagger = (-i) \left( \frac{\partial \hat{U}_T^\dagger}{ \partial \omega} \right) \hat{U}_T \at[]{\omega =\omega_0}$ is a hermitian operator from the fact that $\hat{U}_T^\dagger \hat{U}_T = \mathbb{I} \Leftrightarrow \left( \frac{\partial \hat{U}_T^\dagger}{ \partial \omega} \right) \hat{U}_T + \hat{U}_T^\dagger \left( \frac{\partial \hat{U}_T}{ \partial \omega} \right)  =0$. This is consistent with the conventional form of the QFI, given by the variance of the generator.

Now, we evaluate the explicit form of $\hat{G}_T$ for general time-dependent dynamics with 
\begin{equation}
\hat{U}_{t \rightarrow t'} = \mathcal{T} e^ {-i\int_t^{t'} dt\hat{H}(t)  },
\end{equation}
where $\hat{U}_T = \hat{U}_{0 \rightarrow T}$ is the special case with $t=0$ and $t'=T$. We note that
\begin{equation}
\begin{aligned}
\hat{G}_T &=  i \hat{U}_T^\dagger \left( \frac{\partial \hat{U}_T}{\partial \omega}\at[]{\omega =\omega_0} \right) \\
&= i \hat{U}_T^\dagger \int _0 ^ T d\tau\mathcal{T}e^{-i \int_\tau^T  \hat{H}(t) dt}(-i) \left( \frac{\partial \hat{H}(\tau)}{\partial \omega} \right) \at[]{\omega =\omega_0} \mathcal{T}e^{-i \int_0^\tau \hat{H}(t) dt}\\
&= \int _0^T d\tau \hat{U}_{0 \rightarrow T}^\dagger \hat{U}_{\tau \rightarrow T} \left( \frac{\partial \hat{H}(\tau)}{\partial \omega} \right) \at[]{\omega =\omega_0} \hat{U}_{0 \rightarrow \tau} \\
&= \int _0^T d\tau \hat{U}_{0 \rightarrow \tau}^\dagger \left( \frac{\partial \hat{H}(\tau)}{\partial \omega} \right) \at[]{\omega =\omega_0} \hat{U}_{0 \rightarrow \tau}.
\end{aligned}
\end{equation}

For the Hamiltonian in Eq.~\eqref{eq:effective_H}, we have
\begin{equation}
\frac{\partial \hat{H}(t)}{\partial \omega} = \an^\dagger\an,
\end{equation}
which is time-independent. We also note that the time-evolution operator given by Eq.~\eqref{eq:unitary_form} reads
\begin{equation}
\U_{0 \rightarrow \tau} = e^{i \chi(\tau)} \hat{R}\left(\frac{\varphi(\tau)}{2}\right) \hat{S}\left( r(\tau) \right)\hat{R}\left(\frac{\theta(\tau)}{2}\right),
\end{equation}
at time $\tau$. From this, we obtain
\begin{equation}
\begin{aligned}
\hat{G}_T &= \int _0^T d\tau \hat{U}_{0 \rightarrow \tau}^\dagger (\an^\dagger\an) \hat{U}_{0 \rightarrow \tau} \\
    &= \int _0 ^ T d\tau \hat{R}^\dagger\left(\frac{\theta(\tau)}{2}\right) \hat{S}^\dagger\left(\frac{r(\tau)}{2}\right)\hat{R}^\dagger\left(\frac{\varphi(\tau)}{2}\right)\an^\dagger\an\hat{R}\left(\frac{\varphi(\tau)}{2}\right) \hat{S}\left(\frac{r(\tau)}{2}\right)\hat{R}\left(\frac{\theta(\tau)}{2}\right)\\
    &=\int_0^T d\tau \left(\cosh 2r(\tau)\an^\dagger\an - \frac{e^{-i\theta(\tau)}\an^2+e^{i\theta(\tau)}\an^{\dagger 2}}{2}\sinh 2r(\tau) \right).
\end{aligned}
\end{equation}
By taking the initial state $\ket{\psi_0} = \ket{0}$, we note that
\begin{equation}
\hat{G}_T \ket{0} = - \int_0^T d\tau e^{i\theta(\tau)} \sinh(2r(\tau)) \left( \frac{\an^{\dagger 2}}{2} \right) \ket{0},
\end{equation}
which leads to
\begin{equation}
\begin{aligned}
\bra{0} \hat{G}_T \ket{0} &= 0, \\
\bra{0} \hat{G}_T^2 \ket{0} &= \int_0^T d\tau \sinh(2r(\tau)) e^{i\theta(\tau)} \int_0^T d\tau'  \sinh(2r(\tau')) e^{-i\theta(\tau')} \bra{0} \left( \frac{\an^2 \an^{\dagger 2}}{4} \right) \ket{0} \\
&= \frac{1}{2} \left | \int_0^T d\tau \sinh(2r(\tau)) e^{i\theta(\tau)}  \right|^2.
\end{aligned}
\end{equation}
Finally, we obtain the closed form of the QFI,
\begin{equation}
\label{eq:closed_form_qfi}
    \mathcal{F}_\omega = 4( \langle \psi_0 | \hat{G}_T^2 | \psi_0 \rangle -  \langle \psi_0 | \hat{G}_T | \psi_0 \rangle^2  ) = 2\left | \int_0^T d\tau \sinh(2r(\tau)) e^{i\theta(\tau)}  \right|^2,
\end{equation}
which corresponds to Eq.~(6) in the main text.

\newpage

\section{Proof of Theorems with finite winding number}\label{section:fintinte_winding}

\subsection{Proof of Theorem 1}
In this section, we provide the detailed proof of Theorem~1:
\addtocounter{theorem}{-5}
\begin{theorem}
\label{thm:thm1}
The scaling of QFI with a fixed winding number $n$ is bounded by 
\[
\mathcal{F}_\omega \leq \frac{1}{\omega^2}(\omega T)^{4n+6}+o(T^{4n+6})
\]
where $o(T^{4n+6})$ grows much slower than $T^{4n+6}$.
\end{theorem}
\begin{proof}
By using Eq.~\eqref{eq:closed_form_qfi}, the upper bound of QFI is given by
    \[\begin{split}\label{eq:QFI_element_bound}
        \mathcal{F}_\omega &= 2 \left|\int_0^T d\tau \sinh 2r (\tau) e^{i\theta(\tau)}\right|^2 \leq  2 \left(\int_0^T d\tau \sinh 2r (\tau)\right)^2,
    \end{split}
    \]
where the right side only depends on the squeezing parameter. Hence, the maximum scaling of QFI is bounded by the maximum scaling of $r$.

In order to show the bound for the fixed winding number, we first show the upper bound of squeezing without increasing the winding number, summarized as the following Lemma:
\begin{lemma}\label{lemma:bound_of_r}
    Suppose that $\varphi(t_k) = 2 k \pi$ for any non-negative integer $k$. Then, the squeezing after time $\Delta_k$ without increasing the winding number (i.e., $\varphi(t_k + \Delta_k) < 2 (k+1) \pi$) is upper bounded as
    \begin{equation}
    \label{eq:supp_Lemma1_bound}
	    \cosh (2r(t_k + \Delta_k)) \leq ( \omega^2 \Delta_k^2 +1) \cosh (2r(t_k)) .
    \end{equation}
\end{lemma}
The proof of Lemma~\ref{lemma:bound_of_r} is provided at the end of the section. 

We provide a proof of Theorem~1 by using Lemma~\ref{lemma:bound_of_r}. For a fixed winding number $n$, let us define $t_k$ as the time points satisfying $\varphi(t_k) = 2 k \pi$ for each $k = 0, 1, \cdots, n-1, n$. As there is a region with non-decreasing $\varphi$, i.e., $\dot\varphi(t) > 0$ when $| \cos\varphi | > \tanh(2r)$, the winding number cannot decrease. This implies that there is a single $t_k$ that satisfies $\varphi(t_k) = 2 k \pi$ for each $k$. Now let us set the time duration $\Delta_k = t_{k+1} - t_k$, during which the winding number remains at $k$. We note that $t_0 = 0$ so that $\cosh(r(t_0)) = 1$ and $\tau = \sum_{k=0}^n \Delta_k$. Then, from the bound from Lemma~\ref{lemma:bound_of_r}, for each time interval $\Delta_k$, the squeezing is bounded as $\cosh 2r(t_k + \Delta_k) \leq (\omega^2 \Delta_k^2 +1) \cosh 2r(t_k) $. By repeating this step for each $k = 0, 1, \cdots, n$, we obtain
\begin{equation}
\begin{aligned}
    \cosh 2r (\tau)  \leq \prod_{k=0}^{n} (\omega^2 \Delta_k^2+1) 
 \leq \left[ \sum_{k=0}^n \frac{\omega^2 \Delta_k^2+1}{n+1} \right]^{n+1} 
 = \left[ \frac{\omega^2}{n+1} \left( \sum_{k=0}^n \Delta_k^2 \right) + 1\right]^{n+1} 
 \leq \left[ \frac{\omega^2}{n+1} \left( \sum_{k=0}^n \Delta_k \right)^2 + 1\right]^{n+1} 
 = \left( \frac{\omega^2 \tau^2}{n+1} + 1\right)^{n+1},
\end{aligned}
\end{equation}
where the second inequality comes from the inequality of arithmetic and geometric means and the third inequality comes from $\sum_i x_i^2 \leq (\sum_i x_i)^2$ for $x_i \geq 0$.

By noting that $\sinh x \leq \cosh x $, we have the following bound of the squeezing parameter for a fixed winding number $n$ and the total evolution time $\tau$,
\begin{equation}
	\sinh 2r(\tau) \leq \left( \frac{\omega^2 \tau^2}{n+1} + 1\right)^{n+1}.
\end{equation}

Consequently, the QFI can be bounded as 
\begin{equation}
        \mathcal{F}_\omega \leq  2 \left(\int_0^T d\tau \sinh 2r (\tau)\right)^2 \leq 2 \left(\int_0^T d\tau   \left( \frac{\omega^2\tau^2}{n+1}+1 \right)^{n+1}\right)^2 = 2 T^2 \left[_2F_1\left(\frac{1}{2}, -(n+1), \frac{3}{2}, -\frac{\omega^2 T^2}{n+1}\right)\right]^2,
\end{equation}
where $_2F_1$ is the hypergeometric function. When we focus only on the leading order term in $T$, we have
\begin{equation}
    \mathcal{F}_\omega \leq \frac{2T^4}{(2n+3)^2} \left(\frac{\omega^2T^2}{n+1}\right)^{2n+1} + o(T^{4n+6}) = c_n T^{4n+6} + o(T^{4n+6}),
\end{equation}
with $c_n = \frac{2 \omega^{4n+2}}{(2n+3)^2 (n+1)^{2n+1}}$, which completes the proof.
\end{proof}

\begin{figure}[b]
\begin{center}
\includegraphics[width=.9\linewidth]{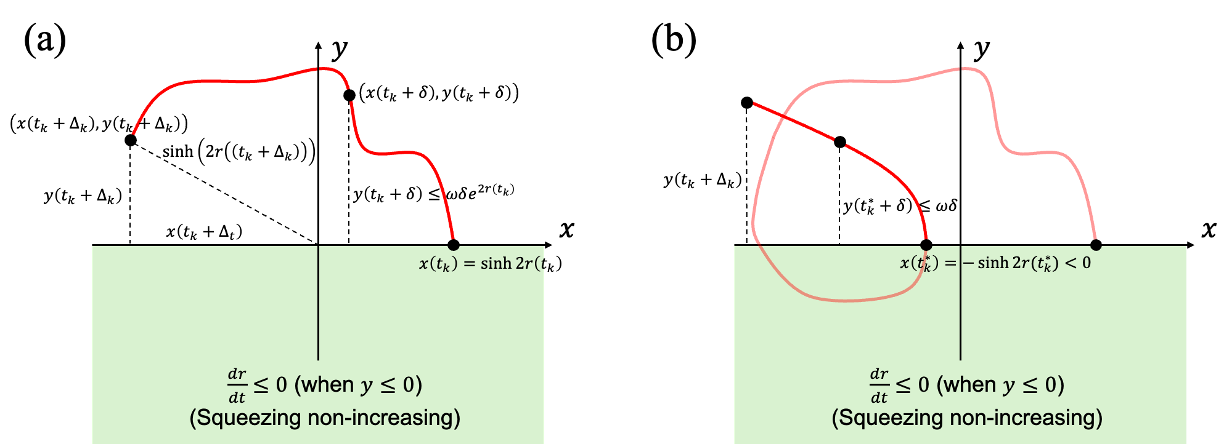}
\caption{Schematic figure of obtaining the bound of the squeezing parameter from time $t_k$ to $t_k + \Delta_k$ without increasing the winding number.}
\label{Fig:supp_bound}
\end{center}
\end{figure}
Now, we provide the proof of Lemma~\ref{lemma:bound_of_r} as follows:
\begin{proof}
From the equation of motion in Eq.~\eqref{eq:Eq_M}, the squeezing parameter $r$ is monotonically decreasing when $\sin\varphi \leq 0$. Therefore, we only need to consider the case with $\sin\varphi \geq 0$. To this end, we parameterize $r$ and $\varphi$ as
\begin{eqnarray}
    x(t) &= \sinh 2r(t) \cos \varphi(t),  \\
    y(t) &= \sinh 2r(t) \sin \varphi(t),
\end{eqnarray}
and focusing on the case with $y(t) \geq 0$. The equation of motion is then rewritten as
\begin{equation}
\begin{aligned}
\label{eq:EOM_xy}
    \dot x &= -(2\omega -\epsilon(t))y,  \\
    \dot y &= (2\omega -\epsilon(t))x +\epsilon(t)\sqrt{1+x^2+y^2}.
\end{aligned}
\end{equation}
We note that at time $t_k$, $y(t_k) = 0$ as $\varphi(t_k) = 2 k \pi$. Consequently, $x(t_k) = \sinh(2r(t_k))$ and $\cosh(2r(t_k)) = \sqrt{x^2(t_k) +1}$. We first show that for any time $\delta$ such that $\varphi(t_k + \delta) < 2 (k+1) \pi$ (see Fig.~\ref{Fig:supp_bound}),
\begin{equation}
\label{eq:supp_y_bound}
y(t_k + \delta) \leq (\omega \delta) e^{2r(t_k)}.
\end{equation}
This can be shown as follows:
\begin{equation}
\begin{aligned}
	\dot y = 2\omega x + \epsilon(t) \left(\sqrt{1+x^2+y^2} -x\right)\leq 2\omega x + \omega \left(\sqrt{1+x^2+y^2} -x\right)= \omega \left( x + \sqrt{1+x^2+y^2} \right),
\end{aligned}
\end{equation}
where $\sqrt{1+x^2+y^2} -x \geq 0$ and $\epsilon \leq \omega$. We then note that $x + \sqrt{1+x^2+y^2} $ is always non-increasing for $y \geq 0$, since
\begin{equation}
    \frac{d}{dt}\left[ x + \sqrt{x^2+y^2+1} \right]= - 2(\omega - \epsilon(t)) y \leq 0.
\end{equation}

Now we have two cases with the trajectories with $y(t)>0$:

(Case 1): The trajectory starts with $x(t_k) = \sinh 2r(t_k)$ and $y = 0$ (see Fig.~\ref{Fig:supp_bound} (a)). In this case, $x + \sqrt{1+x^2+y^2} \leq x(t_k) + \sqrt{1+x^2(t_k) + y^2(t_k)} = \sinh(2r(t_k))+ \sqrt{1+\sinh^2(2r(t_k))} = \sinh(2r(t_k))+ \cosh(2r(t_k)) = e^{2r(t_k)}$. Hence, we have
\begin{equation}
	\dot y \leq \omega e^{2r(t_k)} \Rightarrow y(t_k + \delta) \leq (\omega \delta)e^{2r(t_k)}.
\end{equation}

(Case 2): The trajectory starts with $x = -\sinh 2r(t_k^*) < 0$ and $y=0$ (see Fig.~\ref{Fig:supp_bound} (b)). In this case, $x + \sqrt{1+x^2+y^2} \leq x(t_k^*) + \sqrt{1+x^2(t_k^*) + y^2(t_k^*)} = -\sinh(2r(t_k^*))+ \sqrt{1+\sinh^2(2r(t_k^*))} = -\sinh(2r(t_k^*))+ \cosh(2r(t_k^*)) = e^{-2r(t_k^*)}$. Hence, we have
\begin{equation}
	\dot y \leq \omega e^{-2r(t_k^*)} \Rightarrow y(t_k^* + \delta) \leq (\omega \delta)e^{-2r(t_k^*)}.
\end{equation}

Also, from Eq.~\eqref{eq:EOM_xy}, we obtain
\begin{equation}
\frac{d \cosh(r)}{dt} = \frac{d}{dt} \sqrt{x^2 + y^2 +1} = \epsilon(t) y \leq \omega y,
\end{equation}
since $y \geq 0$ and $\epsilon(t) \leq \omega$. By combining with Eq.~\eqref{eq:supp_y_bound}, this leads to 
\begin{equation}
\begin{aligned}
\cosh(r(t_k + \Delta_k)) - \cosh(r(t_k)) & \leq \int_0^{\Delta_k} \omega y(t_k + \delta) d\delta  \leq  \int_0^{\Delta_k}  (\omega^2 \delta) e^{2r(t_k)} d\delta = \frac{\omega^2 \Delta_k^2}{2}e^{2r(t_k)} \leq \omega^2 \Delta_k^2 \cosh(2r(t_k)) \\
\cosh(r(t_k^* + \Delta_k)) - \cosh(r(t_k^*)) & \leq \int_0^{\Delta_k} \omega y(t_k + \delta) d\delta  \leq  \int_0^{\Delta_k}  (\omega^2 \delta) e^{2r(t_k^*)} d\delta = \frac{\omega^2 \Delta_k^2}{2}e^{-2r(t_k^*)} \leq \omega^2 \Delta_k^2 \cosh(2r(t_k)),
\end{aligned}
\end{equation}
where the last inequality comes from $e^x \leq 2 \cosh x $ and $e^{-y} \leq 1 \leq \cosh x$ for any $x,y \geq 0$. This implies that the squeezing parameter $\cosh(2r)$ cannot increase more than $\omega^2 \Delta_k^2 \cosh(2r(t_k))$ from time $t_k$ to $t_k + \Delta_k$ in any case, which completes the proof of Lemma~\ref{lemma:bound_of_r}.
\end{proof}

\subsection{Proof of Theorem 2}
We provide a proof of Theorem~2 in the main text:
    \begin{theorem}\label{theorem:mono_increasing_g}
    For any control with monotonically increasing $\epsilon(t)$, the QFI scaling cannot exceed $T^{6}$ as the winding number remains zero.
\end{theorem}
\begin{proof}
We prove by contradiction that the winding number remains zero for any monotonically increasing $\epsilon(t)$. Let us assume that the trajectory of a monotonically increasing $\epsilon(t)$ can have winding number $n\geq 1$.
\begin{figure}[b]
\begin{center}
\includegraphics[width=.4\linewidth]{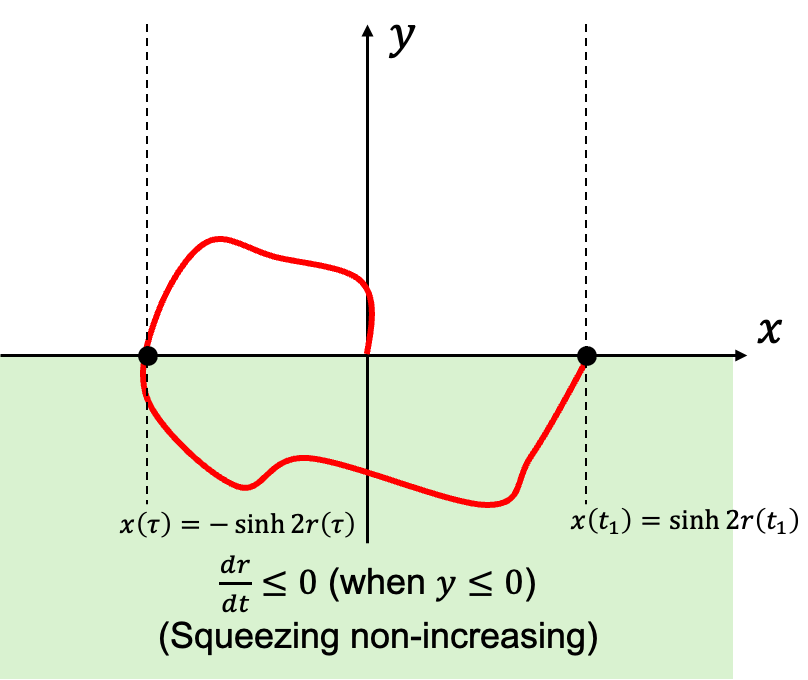}
\caption{Schematic figure of obtaining the bound of the squeezing parameter from time $t_k$ to $t_k + \Delta_k$ without increasing the winding number.}
\label{Fig:supp_bound2}
\end{center}
\end{figure}
Then, for the first winding occurring at time $t_1$, i.e., $(x(t_1), y(t_1)) = (\sinh(2r(t_1)), 0)$, there exists a time point $\tau < t_1$ such that $(x(\tau), y(\tau)) = (-\sinh(2r(\tau)), 0)$ and $y(t) \leq 0$ for any $\tau \leq t \leq t_1$ (see Fig.~\ref{Fig:supp_bound2}). 

Now, let us consider the quantity,
    \[
    C(t) =(2\omega-\epsilon(\tau))\sqrt{x^2(t)+y^2(t)+1}+\epsilon(\tau)x(t).
    \]
By taking the time derivative, we obtain
    \[\label{eq:C_after_t0}
    \frac{dC(t)}{dt} = 2(\epsilon(t) - \epsilon(t_0))\omega y(t) \leq 0 ~\quad{\rm for~} \tau \leq t \leq t_1,
    \]
since $\epsilon$ is monotonically increasing, i.e., $\epsilon(t) - \epsilon(t_0) \geq 0$ and $y(t) \leq 0$ for any $\tau \leq t \leq t_1$. This implies that $C(t)$ is monotone decreasing for $\tau \leq t \leq t_1$, which leads to 
    \[\label{eq:bound_C_t0}
    C(\tau)\geq C(t_1).
    \]

On the other hand, one can directly evaluate $C(t)$ at time $\tau$ and $t_1$ as,
    \begin{equation}
    \begin{aligned}
            C(\tau) &= (2\omega - \epsilon(\tau)) \sqrt{x^2(\tau) +1} + \epsilon(\tau) x(\tau) \\
            &= (2\omega - \epsilon(\tau)) \cosh(2r(\tau)) - \epsilon(\tau) \sinh(2r(\tau)) \\
            &= 2\omega \cosh(2r(\tau)) - \epsilon(\tau) e^{2r(\tau)},
    \end{aligned}
    \end{equation}
and
    \begin{equation}
    \begin{aligned}
            C(t_1) &= (2\omega - \epsilon(\tau)) \sqrt{x^2(t_1) +1} + \epsilon(\tau) x(t_1) \\
            &= (2\omega - \epsilon(\tau)) \cosh(2r(t_1)) + \epsilon(\tau) \sinh(2r(t_1)) \\
            &= 2\omega \cosh(2r(t_1)) - \epsilon(\tau) e^{-2r(t_1)}.
    \end{aligned}
    \end{equation}
We then obtain
    \begin{equation}
    \label{eq:bound_C_other}
    \begin{aligned}
        C(\tau) - C(t_1) &= 2\omega(\cosh(2r(\tau)) - \cosh(2r(t_1)) - \epsilon(\tau) \left(e^{2r(\tau)} - e^{-2r(t_1)} \right)\\
        &= 2\omega(\cosh(2r(\tau)) - \cosh(2r(t_1)) - \epsilon(\tau) \left(e^{2r(\tau)} + e^{-2r(\tau)} - e^{-2r(\tau)} - e^{-2r(t_1)} \right) \\
        &\leq 2\omega(\cosh(2r(\tau)) - \cosh(2r(t_1)) - \epsilon(\tau) \left(e^{2r(\tau)} + e^{-2r(\tau)} - e^{2r(t_1)} - e^{-2r(t_1)} \right) \\
        & = 2\omega(\cosh(2r(\tau)) - \cosh(2r(t_1)) - 2\epsilon(\tau) (\cosh(2r(\tau)) - \cosh(2r(t_1)) \\
        & = 2 (\omega -\epsilon(\tau)) (\cosh(2r(\tau)) - \cosh(2r(t_1))\\
        & \leq 0,
    \end{aligned}
    \end{equation}
where the first inequality comes from $e^{-x} \leq 1 \leq e^y $ for any $x, y \geq 0$, and the second inequality comes from that $\epsilon(t) \leq \omega$ and $r(\tau) \geq r(t_1)$ as the squeezing cannot increase when $y(t) \leq 0$ (as we shown in the proof of Lemma~\ref{lemma:bound_of_r}).

This contradicts with Eq.~\eqref{eq:bound_C_t0} unless $C(\tau) = C(t_1)$. However, the inequality in Eq.~\eqref{eq:bound_C_other} only holds for $r(\tau) = 0 = r(t_1)$, which is a trivial case without any squeezing. As a consequence, there do not exist non-trivial time points $\tau$ and $t_1$ satisfying both Eqs.~\eqref{eq:bound_C_t0} and \eqref{eq:bound_C_other}, which completes the proof by contradiction.
\end{proof}

\subsection{Proof of Theorem 4}
In this section, we prove Theorem~4 in the main text. Let us state Theorem~4 again here.
\addtocounter{theorem}{1}
\begin{theorem}\label{theorem:fixed_n_withouT_criticality}
    For $0\leq \epsilon(t) \leq \epsilon_{\rm max} < \omega$ and a fixed winding number $n$, the squeezing $r$ is upper bounded as
    \begin{equation}
        \sinh 2r(T)  \leq  \frac{1}{(1-(\epsilon_{\rm max}/\omega))^{n+1}},
    \end{equation}
regardless of the total evolution time $T$. This implies that the QFI eventually saturates to $T^2$.
\end{theorem}

\begin{proof}
We prove the statement, similarly to Theorem~\ref{thm:thm1}, by starting with the following lemma:
\begin{lemma}\label{lemma:bounded_r_2}
    Suppose that $\varphi(t_k) = 2 k \pi$ for any non-negative integer $k$. Then, for any protocol with $\epsilon(t) \leq \epsilon_{\rm max} < \omega$, the squeezing after time $\Delta_k$ without increasing the winding number (i.e., $\varphi(t_k + \Delta_k) < 2 (k+1) \pi$) is upper bounded as
    \begin{equation}
    \label{eq:supp_Lemma1_bound2}
	    \cosh (2r(t_k + \Delta_k)) \leq \left( \frac{1}{1-(\epsilon_{\rm max}/\omega)} \right) \cosh (2r(t_k)) .
    \end{equation}
\end{lemma}
The proof of Lemma~\ref{lemma:bounded_r_2} is provided at the end of the section. The main difference from Lemma~\ref{lemma:bound_of_r} is that the bound does not depend on $\Delta_k$. This directly leads to the proof of Theorem~\ref{theorem:fixed_n_withouT_criticality}, by taking the total evolution time $T = \sum_k \Delta_k$ and following the same logic as the proof of Theorem~\ref{thm:thm1} as,
\begin{equation}
    \sinh(2r(T)) \leq \cosh(2r(T)) \leq \left[ \prod_{k=0}^n \left( \frac{1}{1-(\epsilon_{\rm max}/\omega)} \right) \right] \cosh(2r(t_0)) = \frac{1}{\left( 1-(\epsilon_{\rm max}/\omega)\right)^{n+1}},
\end{equation}
where $\cosh(2r(t_0)) = 1$ as $r(t_0) = 0$ at the initial point $t_0 = 0$.
\begin{figure}[b]
\begin{center}
\includegraphics[width=.9\linewidth]{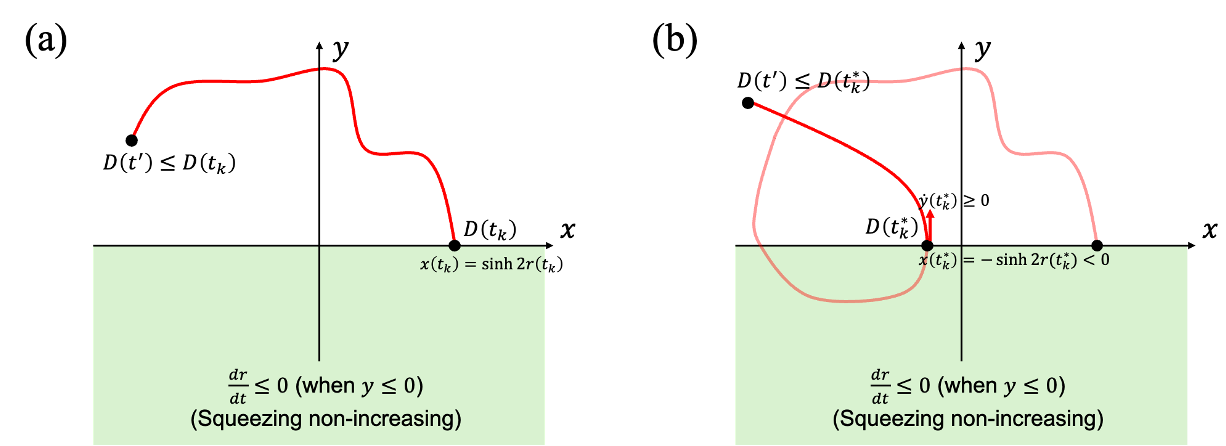}
\caption{Schematic figure for the proof of Lemma~\ref{lemma:bounded_r_2}.}
\label{Fig:supp_bound3}
\end{center}
\end{figure}

Now let us prove Lemma~\ref{lemma:bounded_r_2}. Similarly to the proof of Lemma~\ref{lemma:bound_of_r}, we only need to focus on the regime $y \geq 0$, as the squeezing is non-increasing when $y \leq 0$, or equivalently $\sin \varphi \leq 0$ (see Eq.~\eqref{eq:Eq_M}). We then define the quantity,
    \[
    D(t) = (2\omega-\epsilon_\mathrm{max})\sqrt{x^2(t)+y^2(t)+1}+\epsilon_\mathrm{max}x(t).
    \]
We note that the time derivative of $D(t)$ for $y \geq 0$ is non-positive, i.e., 
    \[\begin{split}
            \frac{dD}{dt}&=(2\omega-\epsilon_\mathrm{max})\frac{d}{dt}\sqrt{x^2+y^2+1}+\epsilon_\mathrm{max}\frac{dx}{dt}\\
            &=(2\omega-\epsilon_\mathrm{max})\epsilon(t)y-\epsilon_\mathrm{max}(2\omega-\epsilon(t))y\\
            &=2\omega(\epsilon(t)-\epsilon_\mathrm{max})y \\
            &\leq 0,
    \end{split}
    \]
as $\epsilon(t)\leq \epsilon_\mathrm{max}$. This implies that $D$ is monotonically decreasing for $y\geq 0$, which leads to
    \[
    \label{eq:D_ineq}
        D(t') \leq D(t),
    \]
for any time $t' \geq t$ as long as $y$ remains non-negative over the time interval $[t,t']$, i.e., the trajectory remains in the upper half plane (see Fig.~\ref{Fig:supp_bound3}).

Now, similarly to the proof of Lemma~\ref{lemma:bound_of_r}, we first consider the case where the trajectory starts with $x(t_k) = \sinh(2r(t_k))$ and $y(t_k)=0$ (see Fig.~\ref{Fig:supp_bound3} (a)). By direct calculation, we obtain
\begin{equation}
\begin{aligned}
    D(t_k) &= (2\omega-\epsilon_\mathrm{max})\sqrt{x^2(t_k)+y^2(t_k)+1}+\epsilon_\mathrm{max}x(t_k)\\
    &= (2\omega - \epsilon_\mathrm{max}) \cosh(2r(t_k)) + \epsilon_\mathrm{max} \sinh(2r(t_k))\\
    &= 2\omega \cosh(2r(t_k)) - \epsilon_\mathrm{max}e^{-2r(t_k)}\\
    &\leq 2\omega \cosh(2r(t_k)),
\end{aligned}
\end{equation}
where $x(t_k) = \sinh(2r(t_k))$ and $y(t_k) = 0$. 

Meanwhile, for $t_k'$, we obtain
\begin{equation}
    \begin{aligned}
        D(t_k') &= (2\omega-\epsilon_\mathrm{max})\sqrt{x^2(t_k')+y^2(t_k')+1}+\epsilon_\mathrm{max}x(t_k')\\
    &= 2(\omega-\epsilon_\mathrm{max})\sqrt{x^2(t_k')+y^2(t_k')+1}+\epsilon_\mathrm{max} \left( \sqrt{x^2(t_k')+y^2(t_k')+1} + x(t_k') \right)\\
    &\geq 2(\omega-\epsilon_\mathrm{max})\cosh(2r(t_k')),
    \end{aligned}
\end{equation}
by noting that $\sqrt{x^2(t_k')+y^2(t_k')+1} = \cosh(2r(t_k'))$ and $\sqrt{x^2(t_k')+y^2(t_k')+1} + x(t_k') \geq 0$.

Then, from Eq.~\eqref{eq:D_ineq}, we have
\begin{equation}
\begin{aligned}
    &2(\omega-\epsilon_\mathrm{max})\cosh(2r(t_k')) \leq D(t_k') \leq D(t_k) \leq 2\omega \cosh(2r(t_k))\\
    &\Rightarrow \cosh(2r(t_k')) \leq \left( \frac{1}{1-(\epsilon_{\rm max}/\omega)} \right) \cosh (2r(t_k)).
\end{aligned}
\end{equation}

For the other case, where the trajectory starts with $x(t_k^*) = -\sinh(2r(t_k^*))$ and $y(t_k^*) = 0$ (see Fig.~\ref{Fig:supp_bound3} (b)), we have $\dot y(t_k^*) \geq 0$ as it crosses the region from $y < 0$ to $y>0$. From the equation of motion (Eq.~\eqref{eq:EOM_xy}), we obtain
\begin{equation}
    \begin{aligned}
        \dot y(t_k^*) = (2\omega -\epsilon(t_k^*))x(t_k^*) +\epsilon(t)\sqrt{1+x^2(t_k^*)+y^2(t_k^*) } = -(\omega - \epsilon(t_k^*))e^{2r(t_k^*)} + \omega e^{-2r(t_k^*)} \geq 0.
    \end{aligned}
\end{equation}
Then, by noting that $\epsilon(t_k^*) \leq \epsilon_{\rm max}$, we have the following bound:
\begin{equation}
  (\omega - \epsilon_{\rm max})e^{2r(t_k^*)}  \leq (\omega - \epsilon(t_k^*))e^{2r(t_k^*)} \leq \omega e^{-2r(t_k^*)}.
\end{equation}
This leads to the bound of $D(t_k^*)$ as
\begin{equation}
\begin{aligned}
    D(t_k^*) &= (2\omega-\epsilon_\mathrm{max})\sqrt{x^2(t_k^*)+y^2(t_k^*)+1}+\epsilon_\mathrm{max}x(t_k^*)\\
    &= (2\omega - \epsilon_\mathrm{max}) \cosh(2r(t_k^*)) - \epsilon_\mathrm{max} \sinh(2r(t_k^*))\\
    &= (\omega - \epsilon_{\rm max})e^{2r(t_k^*)} + \omega e^{-2r(t_k^*)}\\
    &\leq 2 \omega e^{-2r(t_k^*)}\\
    &\leq 2\omega \cosh(2r(t_k)), 
\end{aligned}
\end{equation}
where the last inequality comes from $e^{-2r(t_k^*)} \leq 1 \leq \cosh(2r(t_k))$. Therefore, for the case with $x(t_k^*) <0$, we also have
\begin{equation}
\begin{aligned}
    &2(\omega-\epsilon_\mathrm{max})\cosh(2r(t_k')) \leq D(t_k') \leq D(t_k^*) 
    \leq 2\omega \cosh(2r(t_k))\\
    &\Rightarrow \cosh(2r(t_k')) \leq \left( \frac{1}{1-(\epsilon_{\rm max}/\omega)} \right) \cosh (2r(t_k)),
\end{aligned}
\end{equation}
which completes the proof.
\end{proof}    

\section{Optimal control for the large squeezing limit}
\subsection{Equations of motion for $r\gg1$} 
For large $r$ limit, $r\gg 1$, we can rewrite Eqs.~\eqref{eq:Eq_theta} and \eqref{eq:Eq_M} as
\begin{equation}
\label{eq:eom_at_r>>1}
\begin{aligned}
\frac{\partial r}{\partial t} &= \epsilon(t) \sin (\varphi/2)\cos (\varphi/2),\\
    \frac{\partial \varphi}{\partial t} &= 2\big(\omega -\epsilon(t) \sin^2 (\varphi/2)\big),\\
    \frac{\partial \theta}{\partial t} &= 0,
    \end{aligned}
\end{equation}
since $\coth (2r) \approx 1$ and $\frac{1}{\sinh(2r)} \approx 0$.  In particular, we note that the control parameter $\epsilon(t)$ can be explicitly written
\begin{equation}
\epsilon(t) = \frac{2\omega -\frac{\partial \varphi}{\partial t}}{2\sin^2 (\varphi/2)},
\end{equation}
as a function only depending on $\varphi$ and $\dot\varphi$. We also note that $ \frac{\partial \varphi}{\partial t} = 2\big(\omega -\epsilon(t) \sin^2 (\varphi/2)\big)\geq 0$, $\varphi(t)$ is a monotonically increasing function of $t$. From this, we can express $t$ as an inverse function of $\varphi$. Then, the squeezing parameter after time $T$ with the total phase $\int_0^T \dot \varphi(t) dt = \Phi = \varphi(T)$ can be rewritten as 
\[\begin{split}
\label{eq:closed_form_of_r}
    r(T, \Phi) &= \int_0^{\Phi}d \varphi \frac{\partial t}{\partial \varphi}
\frac{\partial r}{\partial t} \\
&=\int_0^{\Phi}d \varphi \frac{\partial t}{\partial \varphi}
\epsilon(t) \sin (\varphi/2)\cos (\varphi/2)\\
&=\int_0^{\Phi}d \varphi \frac{\partial t}{\partial \varphi}
\left( \frac{2\omega -\frac{\partial \varphi}{\partial t}}{2\sin^2 (\varphi/2)} \right) \sin (\varphi/2)\cos (\varphi/2) \\
&= \int_0^{\Phi}d \varphi \left(2\omega\frac{\partial t}{\partial \varphi} -1 \right)\frac{\cot(\varphi/2)}{2},
\end{split}
\]
where the condition for the total evolution time $\int_0^T dt = T$ is given by
\begin{equation}
\label{eq:norm_cond}
T = \int_0^T dt = \int_0^\Phi d\varphi \frac{\partial t}{\partial \varphi}.
\end{equation}
We also note that the boundary condition $0 \leq \epsilon(t) \leq 1$ becomes
\begin{equation}
\label{eq:boundary_cond}
\frac{1}{2\omega}\leq \frac{\partial t}{\partial \varphi}\leq \frac{\sec^2(\varphi/2)}{2\omega}.
\end{equation}

\subsection{Optimal control protocol for given $T$ and $\Phi$}
We show that the optimal control protocol of $\epsilon$ is given by the on/off control given as the following Lemma:
\begin{lemma}\label{proposition:the_ground_state}
    For a given total evolution time $T$ and total phase $\Phi$, a maximum value of $r(T, \Phi)$ is achievable only by the following on/off process:
    \[\label{eq:the_ground_state_given_Phi_f}
    \frac{\partial t}{\partial \varphi}=  f(\varphi)= \begin{cases} 
        {\frac{\sec^2(\varphi/2)}{2\omega}}, & \mathrm{for }\;\,\cot(\varphi/2)> \cot(\phi^T/2) \\
		\frac{1}{2\omega} , & \mathrm{for }\; \,\cot(\varphi/2)\leq\cot(\phi^T/2),
     \end{cases}
    \]
    which is equivalent to 
        \[\label{eq:d_state_given_Phi_f}
    \epsilon(t) = \begin{cases} 
        \omega, & \mathrm{for }\;\,\cot(\varphi(t)/2)> \cot(\phi^T/2) \\
		0 , & \mathrm{for }\; \,\cot(\varphi(t)/2)\leq\cot(\phi^T/2),
     \end{cases}
    \]
    where $\phi^T$ is defined such that $\frac{\partial t}{\partial \varphi}$ satisfies the normalization condition, $T = \int_0^\Phi d\varphi \frac{\partial t}{\partial \varphi}$.
\end{lemma}
\begin{proof}
We prove the statement by contradiction. Let us assume that the function, $f(\varphi)$ from Eq.~\eqref{eq:the_ground_state_given_Phi_f} does not provide the maximum value. This indicates that there is another function, $g(\varphi)\neq f(\varphi)$, such that
\begin{equation}
\label{eq:e_g-e_f}
r(T, \Phi)\vert_{\frac{dt}{d\varphi} = g(\varphi)} > r(T, \Phi)\vert_{\frac{dt}{d\varphi} = f(\varphi)}.
\end{equation}
As any function $\frac{\partial t}{\partial\varphi} = g(\varphi)$ should obey the boundary condition in Eq.~\eqref{eq:boundary_cond}, and by noting that $f(\varphi)$ is either the maximum or the minimum of the bound, we obtain
\[\label{eq:g-f}
\Delta(\varphi) = g(\varphi)-f(\varphi)\begin{cases} 
        \leq 0, & \mathrm{for }\;\,\cot(\varphi/2)> \cot(\phi^T/2) \\
		\geq 0 , & \mathrm{for }\; \,\cot(\varphi/2)\leq\cot(\phi^T/2).
     \end{cases}
\]
With the normalization condition in Eq.~\eqref{eq:norm_cond}, i.e, $\int_0^\Phi d\varphi f(\varphi) = T = \int_0^\Phi d\varphi g(\varphi)$, Eq.~\eqref{eq:g-f} yields
\[
\int_0^\Phi \Delta(\varphi) d\varphi =  \int_{\cot(\varphi)> \cot(\phi^T/2)} \Delta(\varphi) d\varphi + \int_{\cot(\varphi)\leq \cot(\phi^T/2)} \Delta(\varphi) d\varphi = 0.
\]
We now rewrite Eq.~\eqref{eq:e_g-e_f} as
\begin{equation}
\label{eq:e_g-e_f_normal}
\begin{aligned}
&\left[ r(T, \Phi)\vert_{\frac{dt}{d\varphi} - g(\varphi)} \right] - \left[ r(T, \Phi)\vert_{\frac{dt}{d\varphi} = f(\varphi)} \right]\\
& \quad  = \left[\int_0^{\Phi}d \varphi \left(2\omega\frac{\partial t}{\partial \varphi}\vert_{\frac{dt}{d\varphi} - g(\varphi)} -1 \right)\frac{\cot(\varphi/2)}{2} \right] - \left[ \int_0^{\Phi}d \varphi \left(2\omega\frac{\partial t}{\partial \varphi}\vert_{\frac{dt}{d\varphi} - f(\varphi)} -1 \right)\frac{\cot(\varphi/2)}{2} \right] \\
& \quad= \omega \int_0^{\Phi} \cot(\varphi/2)  \Delta(\varphi) d \varphi\\
& \quad= \omega \left[ \int_{\cot(\varphi/2)> \cot(\phi^T/2)} \cot(\varphi/2) \Delta(\varphi) d\varphi+\int_{\cot(\varphi/2)\leq \cot(\phi^T/2)} \cot(\varphi/2)  \Delta(\varphi) d\varphi \right]\\
& \quad= \omega \left[ \int_{\cot(\varphi/2)> \cot(\phi^T/2)} \cot(\varphi/2) \Delta(\varphi) d\varphi+\int_{\cot(\varphi/2)\leq \cot(\phi^T/2)} \cot(\varphi/2) \Delta(\varphi) d\varphi - \cot(\phi^T/2)\int_0^{\Phi}\Delta(\varphi)d\varphi \right] \\
& \quad= \omega \left[ \int_{\cot(\varphi/2)> \cot(\phi^T/2)} (\underbrace{\cot(\varphi/2)-\cot(\phi^T/2)}_{\geq 0}) \underbrace{\Delta(\varphi)}_{\leq 0} d\varphi +\int_{\cot(\varphi/2)\leq \cot(\phi^T/2)}  (\underbrace{\cot(\varphi/2)-\cot(\phi^T/2)}_{\leq 0}) \underbrace{\Delta(\varphi)}_{\geq 0} d\varphi \right]\\
& \quad\leq 0,
\end{aligned}
\end{equation}
which is a contradiction with Eq.~\eqref{eq:e_g-e_f}. 
\end{proof}

\subsection{Optimal squeezing for a given winding number $n$}
Now let us turn our focus to evaluating the optimal value of $\Phi$ that yields the maximum value of $r(T,\Phi)$ for the given total time $T$. Let us first find the optimal value for a fixed winding number $n$, i.e.,  for the cases where $\Phi_n = 2n \pi +\tilde{\phi}_n$ with $\tilde{\phi}_n < 2\pi$. For this case, the optimal control protocol with the fixed winding number $n$, $\frac{dt}{d\varphi} = f_n(\varphi)$, given by Lemma~\ref{proposition:the_ground_state} leads to
\[
\label{eq:local_sol}
\frac{\partial t}{\partial \varphi}
=f_n(\varphi) = \begin{cases} 
        \frac{\sec^2(\varphi/2)}{2\omega} & \mathrm{for }\;\,2m\pi \leq \varphi < 2m\pi+\phi_n \\
		\frac{1}{2\omega} , & \mathrm{for }\; \,2m\pi +\phi_n \leq \varphi <2(m+1)\pi\\
        \frac{\sec^2(\varphi/2)}{2\omega},  & \mathrm{for }\;\,2n\pi \leq \varphi < 2n \pi + \tilde{\phi}_n,
     \end{cases}
\]
where $0\leq m< n$ is an integer. The maximum squeezing $r(T,\Phi_n)$ obtained from the protocol can be directly evaluated as
\[\begin{split}
\label{eq:E_G(T,n)}
    r(T,\Phi_n)&=\int_{0}^{\Phi_n} \left(2\omega f(\varphi)-1 \right)\frac{\cot(\varphi/2)}{2} d\varphi \\
    &=n\int_{0}^{\phi_n} (\sec^2(\varphi/2)-1)\frac{\cot(\varphi/2)}{2}d\varphi+\int_{0}^{\tilde{\phi}_n} (\sec^2(\varphi/2)-1)\frac{\cot(\varphi/2)}{2}d\varphi\\
    &=n\int_{0}^{\phi_n} \frac{\tan(\varphi/2)}{2}d\varphi+\int_{0}^{\tilde{\phi}_n} \frac{\tan(\varphi/2)}{2}d\varphi \\
    &=-n\ln (\cos (\phi_n/2))-\ln (\cos (\tilde{\phi}_n/2)).
\end{split}
\]
We also note that the normalization condition becomes
\[\label{eq:condition_Phi_n}
\begin{split}
    T &=\int_{0}^{\Phi_n}f_n(\varphi) d\varphi \\
    &=n\Big(\int_{0}^{\phi_n}\frac{\sec^2(\varphi/2)}{2\omega}d\varphi +  \int_{\phi_n}^{2\pi}\frac{1}{2\omega}d\varphi\Big)+\int_{0}^{\tilde{\phi}_n}\frac{\sec^2(\varphi/2)}{2\omega}d\varphi \\
    &=\frac{n}{\omega}\tan(\phi_n/2)+\frac{n}{\omega}(\pi-\phi_n/2)+\frac{1}{\omega}\tan(\tilde{\phi}_n/2).
\end{split}
\]
Hence, to find the optimal $\Phi_n^T = \arg\max_{\Phi_n} r(T,\Phi_n)$ that maximizes Eq.~\eqref{eq:E_G(T,n)} for given $T$, we take the derivatives with respect to $\phi_n$ to meet the following condition:
\[\label{eq:par_div_E_G}
\frac{\partial r(T, \Phi_n)}{\partial \phi_n}=\frac{n}{2} \tan (\phi_n/2)+\frac{1}{2}\tan(\tilde{\phi}_n/2)\frac{\partial \tilde{\phi}_n}{\partial \phi_n}=0.
\]
Meanwhile, by taking the derivatives of the normalization condition in Eq.~\eqref{eq:condition_Phi_n}, we obtain
\[
0=\frac{\partial T}{\partial \phi_n}=\frac{n}{2\omega}(\sec^2(\phi_n/2)-1)+\frac{1}{2\omega}\sec^2(\tilde{\phi}_n/2)\frac{\partial \tilde{\phi}_n}{\partial \phi_n},
\]
which leads to 
\begin{equation}
\frac{\partial \tilde{\phi}_n}{\partial \phi_n} = -\frac{n(\sec^2(\phi_n/2)-1)}{\sec^2(\tilde{\phi}_n/2)} = -n\cos^2(\tilde{\phi}_n/2) \tan^2(\phi_n/2).
\end{equation}
By substituting this expression to Eq.~\eqref{eq:par_div_E_G}, we obtain
\begin{equation}
\begin{aligned}
\label{eq:condition_Phi_n'}
&\frac{1}{2} \left[ n \tan (\phi_n/2) + \tan(\tilde{\phi}_n/2) \left( -n\cos^2(\tilde{\phi}_n/2) \tan^2(\phi_n/2) \right) \right] =0 \\
& \Rightarrow \tan(\phi_n/2) \left( 1 - \sin(\tilde{\phi}_n/2)\cos(\tilde{\phi}_n/2) \tan(\phi_n/2) \right) =0 \\
& \Rightarrow \sin \tilde \phi_n = \frac{2}{\tan(\phi_n/2)}.
\end{aligned}
\end{equation}
Within the range $0 \leq \tilde \phi_n \leq 2\pi$, we obtain the solution satisfying Eq.~\eqref{eq:condition_Phi_n'} as
\begin{equation}
\tilde \phi_n^T = \pi - \sin^{-1}\left( \frac{2}{\tan(\phi_n^T/2)} \right).
\end{equation}
This leads to the optimal phase
\begin{equation}
    \Phi_n^T = 2n\pi + \tilde\phi_n^T = (2n+1)\pi - \sin^{-1}\left( \frac{2}{\tan(\phi_n^T/2)} \right),
\end{equation}
where the angle $\phi_n^T$ corresponding to this optimal control protocol is given by solving the normalization condition in Eq.~\eqref{eq:condition_Phi_n} as
\begin{equation}
    T =\frac{n}{\omega}\tan(\phi_n^T/2)+\frac{n}{\omega}(\pi-\phi_n^T/2)+\frac{1}{\omega} \tan\left[\frac{\pi}{2} - \frac{1}{2} \sin^{-1} \left( \frac{2}{\tan(\phi_n^T/2)}\right) \right] .
\end{equation}

The explicit control protocol is composed of two periods. Eq.~\eqref{eq:local_sol} indicates that $\frac{\partial t}{\partial \varphi}$ has to be $\frac{\sec^2 (\varphi/2)}{2\omega}$ or $\frac{1}{2\omega}$, which correspond to $\epsilon(t)=\omega$ or $\epsilon(t)=0$.  By integrating $\frac{\partial t}{\partial \varphi}$ for each period, we obtain the explicit time interval for turning on the squeezing (i.e., $\epsilon(t) = \omega$) as follows:
\[
\Delta T_\mathrm{ON} = \int_{2 m \pi}^{2 m \pi +\phi_n^T} d\varphi\frac{\partial t}{\partial \varphi} =\int_{2 m \pi}^{2 m \pi +\phi_n^T} d\varphi\frac{\sec^2 (\varphi/2)}{2\omega} = \frac{1}{\omega}\tan (\phi^T_n /2),
\]
for each $m$. Similarly, the time interval for turning off $\Delta T_\mathrm{OFF}$ (i.e., $\epsilon(t)=0$) is obtained as
\[
\Delta T_\mathrm{OFF} = \int_{2 m \pi+\phi_n^T}^{2(m+1)\pi} d\varphi\frac{\partial t}{\partial \varphi} = \int_{2m\pi+\phi_n^T}^{2(m+1)\pi} d\varphi\frac{1}{2\omega} = \frac{1}{\omega}(\pi - \phi^T_n /2).
\]

\section{Proof of Theorems with exponential scaling}\label{section:exponential}
\subsection{Proof of Theorem~3}
In this section, we prove Theorem~3 in the main text. Let us state Theorem~3 again here.
\addtocounter{theorem}{-2}
\begin{theorem}
\label{theorem:exponetial_g=1} 
The scaling bound of the QFI for a long-time limit ($T \gg 1$) is given by
\begin{equation}
\mathcal{F}_\omega(T) \propto e^{\Gamma \omega T}
\end{equation}
with $\Gamma \approx 0.9745$. This bound can be saturated by taking the winding number $n \approx 0.169 \omega T$.
\end{theorem}
\begin{proof}
In order to prove Theorem~3, we find the global optimum of the squeezing for a given time $T$, i.e.,
\[
r_\mathrm{max}(T)=\max_{n} r(T,\Phi_n^T),
\]
and identify the optimal winding number in the limit of $T \gg 1$. By considering the local optimal values of Eq.~\eqref{eq:E_G(T,n)}, we have 
\begin{equation}
    \frac{r(T,\Phi_n^T)}{\omega T} = - \left( \frac{n}{\omega T} \right) \ln(\cos(\phi_n^T/2)) - \left( \frac{1}{\omega T} \right) \ln(\cos(\tilde\phi_n/2)) \overset{T \gg 1 }{\approx}  -\left( \frac{n}{\omega T} \right) \ln(\cos(\phi_n^T/2)),
\end{equation}
in the limit of $n \gg 1$ and $\omega T \gg 1$. Also, the boundary condition for the total evolution time given in Eq.~\eqref{eq:condition_Phi_n} can be rewritten as 
\begin{equation}
    \label{eq:boundary_cond_large}
    1 = \left(\frac{n}{\omega T} \right) \tan(\phi_n^T/2) + \left(\frac{n}{\omega T} \right)(\pi-\phi_n^T/2) + \frac{1}{\omega T} \tan(\tilde \phi_n /2) \overset{T \gg 1 }{\approx}  \left(\frac{n}{\omega T} \right) \left[ \tan(\phi_n^T/2) + (\pi-\phi_n^T/2) \right].
\end{equation}
Here, the terms related to the trajectories associated with the last winding number (i.e., $\tilde \phi_n$) can be ignored in the limit of $n \gg 1$ and $\omega T \gg 1$, as they contribute only polynomially to the QFI, as shown in Lemma~\ref{lemma:bound_of_r}. Also, in this limit, $\frac{n}{\omega T}$ can be regarded as continuous. This allows us to rewrite the squeezing as a function of $\phi_n^T$,
\[
\label{eq:opt_r_exp_large_T}
r(T, \Phi_n^T) \overset{T \gg 1 }{\approx}  - \frac{ \omega T \ln (\cos (\phi_n^T/2))}{\tan(\phi_n^T/2)+(\pi-\phi_n^T/2)}.
\]
We then find the optimal $\phi_n^T$ as
\[\label{eq:condition_min_E_G(x)}
\frac{d r(T, \Phi_n^T)}{d\phi_n^T}=\omega T \frac{\tan(\phi_n^T/2)\left(\tan(\phi_n^T/2) +\pi - \phi_n^T/2 \right)+(\sec^2(\phi_n^T/2)-1)\ln(\cos(\phi_n^T/2)) }{2\left(\tan(\phi_n^T/2)+\pi-\phi_n^T/2\right)^2}=0,
\]
which leads to the following condition:
\[
\begin{aligned}
\tan(\phi_n^T/2)(1+\ln(\cos(\phi_n^T/2)))+\pi-\phi_n^T/2=0.
\end{aligned}
\]
The solution of the above equation is obtained at $\phi_{n, {\rm opt}}^T \approx 2.664$. From the boundary condition in Eq.~\eqref{eq:boundary_cond_large}, we have
\begin{equation}
    n_{\rm opt} \overset{T \gg 1 }{\approx}  \frac{\omega T}{\tan(\phi_{n, {\rm opt}}^T/2)+\pi-\phi_{n, {\rm opt}}^T/2} \approx 0.1690 (\omega T),
\end{equation}
which linearly increases by the total evolution time $T$. Finally, the squeezing parameter for this solution is obtained from Eq.~\eqref{eq:opt_r_exp_large_T} as
\begin{equation}
    r_{\rm max}(T) \overset{T \gg 1 }{\approx}  \frac{\omega T}{\tan(\phi_{n, {\rm opt}}^T/2)} \approx 0.2436 (\omega T).
\end{equation}
As the QFI scales exponentially with $r_{\rm max}(T)$, i.e.,
\begin{equation}
    {\cal F}_\omega(T) \propto e^{4r_{\rm max}(T)} = e^{ \Gamma \omega T},
\end{equation}
we have $\Gamma \approx 4 \times 0.2436 = 0.9745$ for large $T$, which completes the proof.
\end{proof}
    
\subsection{Proof of Theorem~5}
We now present the proof of Theorem~5 in the main manuscript:
\addtocounter{theorem}{1}
\begin{theorem}\label{theorem:expon_g<1}
With the control parameter in the range $0 \leq \epsilon(t) \leq \epsilon_\mathrm{max}$, the fundamental scaling limit is given as
\begin{equation}
\mathcal{F}_\omega(T) \propto e^{\Gamma(\epsilon_\mathrm{max}) \omega T}.
\end{equation}
\end{theorem}
\begin{proof}
We follow a similar argument to that in the proof of Theorem~\ref{theorem:exponetial_g=1}, by starting with the same equation for $r(\Phi)$ given in Eq.~\eqref{eq:closed_form_of_r},
\[
\begin{split}
        r(\Phi)= \int_0^{\Phi}d \varphi \left(2\omega f(\varphi) -1 \right)\frac{\cot(\varphi/2)}{2},
\end{split}
\]
with $\frac{\partial t}{\partial \varphi}=  f(\varphi)$ satisfying the same normalization condition,
\[
T =\int_0^\Phi f(\varphi)d\varphi.
\]
Meanwhile, the boundary condition of $f(\varphi)$ is modified in terms of $\epsilon_\mathrm{max}$ as
\[\label{eq:boundary_cond_f}
\frac{1}{2\omega}\leq f(\varphi)\leq \frac{1}{2(\omega-\epsilon_\mathrm{max}\sin^2(\varphi/2))}.
\]
The optimal solution for a given winding number $n$ is written as
\[
f_n(\varphi) = \begin{cases} 
       \frac{1}{2(\omega-\epsilon_\mathrm{max}\sin^2(\varphi/2))}, & \mathrm{for }\;\,2m\pi \leq \varphi < 2m\pi+\phi_n^T \\
		\frac{1}{2\omega} , & \mathrm{for }\; \,2m\pi +\phi_n^T \leq \varphi <2(m+1)\pi\\
        \frac{1}{2(\omega-\epsilon_\mathrm{max}\sin^2(\varphi/2))},  & \mathrm{for }\;\,2n\pi \leq \varphi < 2n\pi+\tilde{\phi}_n^T 
     \end{cases},
\]
with the maximum value of $r(T,\Phi_n^T)$,
\[\begin{split}
    r(T,\Phi_n^T)&=n\int_{0}^{\phi_n^T} \Big(\frac{1}{1-(\epsilon_\mathrm{max}/\omega)\sin^2(\varphi/2)}-1\Big)\frac{\cot(\varphi/2)}{2}d\varphi+\int_{0}^{\tilde{\phi}_n^T} \Big(\frac{1}{1-(\epsilon_\mathrm{max}/\omega)\sin^2(\varphi/2)}-1\Big)\frac{\cot(\varphi/2)}{2}d\varphi\\
    &=-\frac{n}{2}\ln(1-(\epsilon_\mathrm{max}/\omega)\sin^2(\phi^T_n/2))-\frac{1}{2}\ln(1-(\epsilon_\mathrm{max}/\omega)\sin^2(\tilde{\phi}^T_n/2)),
\end{split}
\]
and the normalization condition
\[
\begin{split}
    T&=\int_{0}^{\Phi}f_n(\varphi) d\varphi
    =n\Big(\int_{0}^{\phi_n^T} \frac{1}{2(\omega-\epsilon_\mathrm{max}\sin^2(\varphi/2))}d\varphi
    +  \int_{\phi_n^T}^{\pi}\frac{1}{2\omega}d\varphi\Big)
    +\int_{0}^{\tilde{\phi}_n^T}\frac{1}{2(\omega-\epsilon_\mathrm{max}\sin^2(\varphi/2))}d\varphi\\
    &=\frac{n}{\omega}\frac{\tan^{-1}(\sqrt{1-(\epsilon_\mathrm{max}/\omega)}\tan(\phi_n^T/2))}{\sqrt{1-(\epsilon_\mathrm{max}/\omega)}}+\frac{n}{\omega}(\pi-\phi_n^T/2)+\frac{1}{\omega}\frac{\tan^{-1}(\sqrt{1-(\epsilon_\mathrm{max}/\omega)}\tan(\tilde{\phi}_n^T/2))}{\sqrt{1-(\epsilon_\mathrm{max}/\omega)}}.
\end{split}
\]

Similarly to the case with $\epsilon_\mathrm{max}=\omega$, we can calculate the optimal scale of $r(T)$. In the limit of $T\gg 1$ we obtain
\begin{equation}
\begin{aligned}
    r(T,\Phi_n^T) &\overset{T \gg 1}{\approx} -\frac{n}{2}\ln(1-(\epsilon_\mathrm{max}/\omega)\sin^2(\phi^T_n/2)) \\    
    T &\overset{T \gg 1}{\approx} \frac{n}{\omega} \left[ \frac{\tan^{-1}(\sqrt{1-(\epsilon_\mathrm{max}/\omega)}\tan(\phi_n^T/2))}{\sqrt{1-(\epsilon_\mathrm{max}/\omega)}} + (\pi-\phi_n^T/2) \right],
\end{aligned}
\end{equation}
which leads to the following expression
\[
r(T, \Phi_n^T)=-\frac{\omega T}{2} \ln(1-(\epsilon_\mathrm{max}/\omega)\sin^2(\phi_n^T/2)) \left[ \frac{\tan^{-1}(\sqrt{1-(\epsilon_\mathrm{max}/\omega)}\tan(\phi_n^T/2))}{\sqrt{1-(\epsilon_\mathrm{max}/\omega)}} + (\pi-\phi_n^T/2)\right]^{-1}.
\]
The maximum value of $r$ for a given total evolution time $T$ is then obtained by 
\begin{equation}
    r_{\rm max}(T, \epsilon_{\rm max}) = (\omega T) \max_{0\leq \phi_n^T\leq \pi} \left[ -\frac{1}{2} \ln(1-(\epsilon_\mathrm{max}/\omega)\sin^2(\phi_n^T/2)) \left[ \frac{\tan^{-1}(\sqrt{1-(\epsilon_\mathrm{max}/\omega)}\tan(\phi_n^T/2))}{\sqrt{1-(\epsilon_\mathrm{max}/\omega)}} + (\pi-\phi_n^T/2)\right]^{-1}  \right].
\end{equation}
The result of the optimization can be obtained numerically, which is given in the main manuscript Fig.~3(a) (inset figure). From this, we obtain the bound of QFI with $\epsilon(t)\leq \epsilon_\mathrm{max}<\omega$,
\[
\mathcal{F}_\omega \propto e^{4r_{\rm max}(T, \epsilon_{\rm max})}= e^{\Gamma(\epsilon_\mathrm{max})\omega T},
\]
where $\Gamma(\epsilon_{\rm max}) = \frac{4r_{\rm max}(T, \epsilon_{\rm max})}{\omega T}$.
\end{proof}

\subsection{QFI scaling under thermal dissipation}
We analyze the scaling behavior of the QFI under thermal dissipation described by the following Lindblad equation,
\[
\label{eq:Lindblad_eq}
\frac{d}{dt}\R=\mathcal{L}(\rho)=i[\R,\h(t)] + \gamma (N_0+1) (\an \R \an^\dagger -\tfrac{1}{2}\{\an^\dagger \an, \R \}) + \gamma N_0 (\an^\dagger \R \an -\tfrac{1}{2}\{\an \an^\dagger, \R \}).
\]
We provide the bound on QFI as follows:
\begin{theorem}\label{theorem:expon_open_dyn}
    For $0 \leq \epsilon(t) \leq \epsilon_\mathrm{max}$, the fundamental scaling limit of QFI is given by
\begin{equation}
    \mathcal{F}_\omega(T) \leq A T e^{ \left( \frac{\Gamma(\epsilon_{\rm max})  \omega }{2} - \gamma 
    \right)T},
\end{equation}
with some coefficient $A$.
\end{theorem}
\begin{proof}
Since both the Hamiltonian evolution and thermal dissipation are Gaussian processes without displacement, the quantum state at each time $t$ is a Gaussian state, completely described by the covariance matrix,
\[\begin{split}
    V&=\begin{pmatrix}
\mean{\hat{x}^2} & \mean{\hat{x}\hat{p}+\hat{p}\hat{x}}/2 \\
\mean{\hat{x}\hat{p}+\hat{p}\hat{x}}/2 & \mean{\hat{p}^2}
\end{pmatrix}\\
&
=\frac{2\bar{N} +1}{2}\begin{pmatrix}
\cosh{2r}-\sinh{2r} \cos{\varphi}  & \sinh{2r}\sin{\varphi} \\
\sinh{2r}\sin{\varphi} & \cosh{2r}+\sinh{2r} \cos{\varphi} 
\end{pmatrix}= \frac{2\bar{N} +1}{2} R(\phi/2) S(2r) R^T(\phi/2),
\end{split}
\]
where 
$S(2r) = \left(\begin{matrix}
    e^{-2r} & 0 \\
    0 & e^{2r}
\end{matrix}
\right)$ with the squeezing $r$, and 
$R(\varphi/2) = \left(\begin{matrix}
    \cos\varphi/2 & \sin\varphi/2 \\
    -\sin\varphi/2 & \cos\varphi/2
\end{matrix}
\right)$ with the rotation angle $\varphi$, and mean photon number $\bar{N}$.
For convenience, we will define $\mu$, the mean photon number, as 
\[
\mu =\frac{2\bar{N} +1}{2}.
\]
The equation of motion of elements in $V$ is given by
\[
\frac{\partial}{\partial t} \begin{pmatrix}
    \mean{\hat{x}^2}\\
    \mean{\hat{p}^2}\\
    \mean{\hat{x}\hat{p}+\hat{p}\hat{x}}/2\\
    1
\end{pmatrix}
=\begin{pmatrix}
    -\gamma & 0 & 2\omega & \gamma\frac{2N_0+1}{2}\\
    0 & -\gamma & 2\epsilon(t)-2\omega & \gamma\frac{2N_0+1}{2}\\
    \epsilon(t)-\omega & \omega & -\gamma & 0\\
    0 & 0 & 0 & 0
\end{pmatrix} \begin{pmatrix}
    \mean{\hat{x}^2}\\
    \mean{\hat{p}^2}\\
    \mean{\hat{x}\hat{p}+\hat{p}\hat{x}}/2\\
    1
\end{pmatrix}.
\]
If we rewrite the time evolution of the covariance in terms of $r$, $\varphi$, and $\mu$, we obatin,
\[
\frac{\partial }{\partial t} \begin{pmatrix}
    \mu\cosh 2r \\
    \mu\sinh 2r \cos \varphi\\
    \mu\sinh 2r \sin \varphi\\
    1
\end{pmatrix}
=\begin{pmatrix}
    -\gamma & 0 & \epsilon(t) & \gamma\frac{2N_0+1}{2}\\
    0 & -\gamma & \epsilon(t)-2\omega & 0\\
    \epsilon(t) & 2\omega -2\epsilon(t) & -\gamma & 0\\
    0 & 0 & 0 & 0
\end{pmatrix} \begin{pmatrix}
    \mu\cosh 2r \\
    \mu\sinh 2r \cos \varphi\\
    \mu\sinh 2r \sin \varphi\\
    1
\end{pmatrix},
\]
which provides the equation of motion,
\begin{eqnarray}
\label{eq:mu_time_derivative}
\frac{\partial \mu}{\partial t} &=& -\gamma \mu + \gamma \cosh 2r \frac{2N_0 +1 }{2},\\
\label{eq:r_time_derivative}
\frac{\partial r}{\partial t} &=& \frac{\epsilon(t)\sin \varphi}{2} -\gamma \sinh 2r \frac{2N_0+1}{4\mu},\\
\label{eq:varphi_time_derivative}
\frac{\partial \varphi}{\partial t} &=& 2\omega - \epsilon(t)(1-\coth 2r \cos \varphi).
\end{eqnarray}
We also note that the QFI is determined by the final state $\hat\rho(T)$ and its derivative $\partial_\omega \hat\rho(T)$ as~\cite{Braunstein.1994}
\[
\mathcal{F}_\omega(\R(T)) = 2\sum_{k,l}\frac{|\bra{k}\partial_\omega \R(T)\ket{l}|^2}{\lambda_k+\lambda_l},
\]
where $\lambda_k$ and $\ket{k}$ are the eigenvalues and eigenstates of $\hat\rho$, respectively. We then show the following bound of the QFI:
\begin{lemma} \label{lemma:qfi_bound_thermal} The QFI for a given total evolution time $T$ is bounded by the time integral of the squeezing at each time, given as,
    \begin{equation}
        \mathcal{F}_\omega(\R(T)) \leq 4T\int_0^T dt \sinh^2 2r(t).
    \end{equation}
\end{lemma}
We will provide the proof of the Lemma at the end of this section.

Let us now derive an upper bound for the squeezing parameter $r(t)$. We note that Eq.~\eqref{eq:varphi_time_derivative} is the same as the dynamics of the closed system. Moreover, we obtain the following equation from Eq.~\eqref{eq:mu_time_derivative} and Eq.~\eqref{eq:r_time_derivative},
\[
\frac{\partial}{\partial t}\mu\sinh 2r=-\gamma \mu \sinh 2r  + \epsilon(t)\mu \sin \varphi \cosh2r.
\]
For the large $r\gg1$ limit, we have asymptotic equations
\begin{eqnarray}
    \frac{\partial s }{\partial t} &=&\epsilon(t)\sin(\varphi/2)\cos(\varphi/2)\label{eq:open_dr_rgg1}\\
    \frac{\partial \varphi}{\partial t}&=&2(\omega-\epsilon(t)\sin^2(\varphi/2)),
    \label{eq:open_dtheta_rgg1}
\end{eqnarray}
by reparameterizing $s=(\gamma t +\ln(\mu \sinh 2r)/2$, which is exactly the same form as the asymptotic equation of a closed system described by Eq.~\eqref{eq:eom_at_r>>1}. Hence, the maximum scaling of $s$ becomes
$$
s_{\rm max} = \frac  {\Gamma(\epsilon_{\rm max})}{4} (\omega T).
$$

In the section about the closed system, we illustrate that the best scaling is given by 
\[
\lim_{t\rightarrow\infty}\frac{s}{t} \leq \frac{\Gamma (\epsilon_\mathrm{max})\omega}{4},
\]
where $\Gamma(\epsilon_\mathrm{max}) \omega/4$ is the function of maximum value of $\epsilon(t)$, which is equal to $0.2436$ typically when $\epsilon_\mathrm{max}=\omega$.
By definition, we have the maximal scaling,
\[\label{eq:open_scaling_lnnr_1}
\lim_{t\rightarrow\infty}\frac{\ln (\mu\sinh 2r)}{t} \leq \frac{\Gamma(\epsilon_\mathrm{max})\omega}{2}-\gamma.
\]

Now we derive the maximum scaling of $r$ from Eq.~\eqref{eq:open_scaling_lnnr_1}. Let us assume that there is some process $\epsilon(t)$ which gives
\[\label{eq:open_scaling_lnnr}
\lim_{t\rightarrow\infty}\frac{\ln (\mu\sinh 2r)}{t} = \Gamma' \leq \frac{\Gamma(\epsilon_\mathrm{max})\omega}{2}-\gamma.
\]
Meanwhile, we note that Eq.~\eqref{eq:mu_time_derivative} can be rewritten as
\[
\label{eq:open_dmu}
\frac{d}{dt}\mu^2=-2\gamma \mu^2 + \mu \gamma\cosh 2r (1+2N_0).
\]
In the limit of $r\gg1$, this equation can be rewritten as
\[
\frac{d}{dt}\mu^2=-2\gamma \mu^2 +\mu\gamma\sinh 2r(1+2N_0).
\]
By solving the first-order differential equation, we obtain that
\[
\mu(t)^2=e^{-2\gamma t}  \int_0^t d\tau  e^{2\gamma \tau} \mu(\tau)\gamma\sinh 2r(\tau)(1+2N_0).
\]
The scaling of $\mu(t)$ is given by
\[
\lim_{t\rightarrow\infty}\frac{\ln (\mu)}{t}=-\gamma+\lim_{t\rightarrow\infty}\frac{1}{2t}\ln\left(\int_0^t d\tau e^{2\gamma \tau} \gamma\mu(\tau)\sinh 2r(\tau)(1+2N_0)\right).
\]
By using the fact that
\[
\lim_{t\rightarrow\infty}\frac{1}{2t}\ln( e^{2\gamma t} \gamma\mu(t)\sinh 2r(t)(1+2N_0){2})=\frac{\Gamma'}{2}+\gamma,
\]
the scaling of $\mu$ is 
\[
\lim_{t\rightarrow\infty}\frac{\ln (\mu)}{t}=\frac{\Gamma'}{2}.
\]
It is obvious that the scaling of the squeeze parameter is equal to
\[
\lim_{t\rightarrow\infty}\frac{r}{t}=\frac{\Gamma'}{4}.
\]
Since $\Gamma'$ is upper bounded by $\Gamma(\epsilon_\mathrm{max})-\gamma$, the maximum scaling of $r$ is 
\[
\lim_{t\rightarrow\infty}\frac{r}{t}\leq\frac{\Gamma(\epsilon_\mathrm{max})\omega}{8}-\frac{\gamma}{4}.
\]
Hence, from Lemma~\ref{lemma:qfi_bound_thermal}, we obtain the bound,
\[
\mathcal{F}_\omega  \leq T\int_0^T dt \mathcal{F}(\rho(t),i[\rho(t),\an^\dagger\an]) \leq 4T\int_0^T dt \sinh^2 2r(t) \approx ATe^{\frac{\Gamma(\epsilon_\mathrm{max})\omega T}{2} -\gamma},
\]
which completes the main proof.

Finally, we prove the bound provided in Lemma~\ref{lemma:qfi_bound_thermal}:
\begin{equation}
        \mathcal{F}_\omega(\R(T)) \leq 4T\int_0^T dt \sinh^2 2r(t).
    \end{equation}
We express the QFI in terms of $\hat\rho(T)$ and $\partial_\omega \hat\rho(T)$ as
\[
\mathcal{F}_\omega(\R(T)) = 2\sum_{k,l}\frac{|\bra{k}\partial_\omega \R(T)\ket{l}|^2}{\lambda_k+\lambda_l} = {\cal B}(\hat\rho, \partial_\omega \hat\rho),
\]
where we define ${\cal B}(\hat\rho, \hat A) = 2\sum_{k,l}\frac{|\bra{k} \hat A \ket{l}|^2}{\lambda_k+\lambda_l}$ with the eigendecomposition $\hat\rho = \sum_{k} \lambda_k \ket{k}\bra{k}$.
Both $\hat\rho(T)$ and $\partial_\omega \hat\rho(T)$ can be directly obtained by the following time-ordered integral,
\begin{eqnarray}
    \R(T) &=& \mathcal{T} e^{\int_0^T dt \mathcal{L}(t)}\pro{0}{0},\\
    \partial_\omega \R(T) &=& \int_0^T dt \mathcal{T} e^{\int_t^T d\tau \mathcal{L}(\tau)} \left(\partial_\omega {\cal L}(t)\right)  = \int_0^T dt \mathcal{T} e^{\int_t^T d\tau \mathcal{L}(\tau)} i[\R(t),\an^\dagger\an].
\end{eqnarray}
However, the presence of the time-ordered integral poses limitations when deriving general properties and upper bounds of the QFI. To overcome these limitations, we introduce a method to derive an upper bound by eliminating the time-ordered integral.

Let us first consider two hypothetical Lindbladian operators, $\mathcal{L}_1(t;\omega)$ and $\mathcal{L}_1(t;\omega)$, where
\begin{eqnarray}\label{eq:L_1andL_2}
    \mathcal{L}_1(t;\omega=\omega_0) &=& \mathcal{L}_2(t;\omega=\omega_0) = \mathcal{L},\\
    \frac{\partial \mathcal{L}_1(t;\omega)}{\partial \omega}\at[]{\omega=\omega_0} (\hat\rho) &=&
    \begin{cases}
        2\frac{\partial \mathcal{L}}{\partial \omega} (\hat\rho)\at[]{\omega=\omega_0}= 2i[\hat\rho, \an^\dagger\an] &\quad \text{for } 0\leq t< \frac{T}{2}  \\
        0&\quad \text{for } \frac{T}{2}\leq t< T
    \end{cases}\\
    \frac{\partial \mathcal{L}_2(t;\omega)}{\partial \omega}\at[]{\omega=\omega_0} (\hat\rho) 
    &=& 
        \begin{cases}
        0 &\quad \text{for } 0\leq t< \frac{T}{2}  \\
        2\frac{\partial \mathcal{L}}{\partial \omega} (\hat\rho)\at[]{\omega=\omega_0} = 2i[\hat\rho, \an^\dagger\an] &\quad \text{for } \frac{T}{2}\leq t< T
    \end{cases}.
\end{eqnarray}
The two hypothetical Lindbladian operators generate the dynamics and the state at $T$,
\begin{eqnarray}
    \R_1(T) &=& \mathcal{T} e^{\int_0^T dt \mathcal{L}_1(t)}\pro{0}{0}=\mathcal{T} e^{\int_0^T dt \mathcal{L}(t)}\pro{0}{0}=\R(T),\\
    \R_2(T) &=& \mathcal{T} e^{\int_0^T dt \mathcal{L}_2(t)}\pro{0}{0}=\mathcal{T} e^{\int_0^T dt \mathcal{L}(t)}\pro{0}{0}=\R(T),\\
    \partial_\omega \R_1(T) &=&  2\int_0^{T/2} dt \mathcal{T} e^{\int_t^T d\tau \mathcal{L}_1(\tau)} i[\R(t),\an^\dagger\an]
    =2\int_0^{T/2} dt \mathcal{T} e^{\int_t^T d\tau \mathcal{L}(\tau)} i[\R(t),\an^\dagger\an],\\
    \partial_\omega \R_2(T) &=&  2\int_{T/2}^T dt \mathcal{T} e^{\int_t^T d\tau \mathcal{L}_2(\tau)} i[\R(t),\an^\dagger\an]
    =2\int_{T/2}^{T} dt \mathcal{T} e^{\int_t^T d\tau \mathcal{L}(\tau)} i[\R(t),\an^\dagger\an],
\end{eqnarray}
where we use the identity from Eq.~\eqref{eq:L_1andL_2}.
The original $\R(T)$ can be recovered as the average of $\R_1(T)$ and $\R_2(T)$, $\R(T) = \tfrac{1}{2}(\R_1(T)+\R_2(T))$. From the convexity property of the QFI, we obtain the inequality,
\[\label{eq:convex_QFI_bi}
\mathcal{F}_\omega(\R(T)) =\mathcal{F}_\omega(\tfrac{1}{2}(\R_1(T)+\R_2(T))) \leq \frac{1}{2}(\mathcal{F}_\omega (\R_1(T))+ \mathcal{F}_\omega (\R_2(T))).
\]
Moreover, we can obtain QFI corresponding to each of $\R_1(T)$ and $\R_2(T)$,
\begin{eqnarray}
    \mathcal{F}_\omega(\R_1(T)) &=& \mathcal{B}\left(\R_1(T),\frac{\partial\R_1(T)} {\partial\omega}\right) \\
    &=& \mathcal{B}\left(\R(T),2\int_{0}^{T/2} dt \mathcal{T} e^{\int_t^T d\tau \mathcal{L}(\tau)} i[\R(t),\an^\dagger\an]\right)\notag \\
    &=&4\mathcal{B}\left(\R(T),\int_{0}^{T/2} dt \mathcal{T} e^{\int_t^T d\tau \mathcal{L}(\tau)} i[\R(t),\an^\dagger\an]\right),\label{eq:QFI_R1} \\
    \mathcal{F}_\omega(\R_2(T)) &=& \mathcal{B}\left(\R_2(T),\frac{\partial\R_2(T)} {\partial\omega}\right) \\
    &=&\mathcal{B}\left(\R(T),2\int_{T/2}^{T} dt \mathcal{T} e^{\int_t^T d\tau \mathcal{L}(\tau)} i[\R(t),\an^\dagger\an]\right)\notag\\
    &=&4\mathcal{B}\left(\R(T),\int_{T/2}^{T} dt \mathcal{T} e^{\int_t^T d\tau \mathcal{L}(\tau)} i[\R(t),\an^\dagger\an]\right)\label{eq:QFI_R2}.
\end{eqnarray}
Combining Eq.~\eqref{eq:convex_QFI_bi}, Eq.~\eqref{eq:QFI_R1}, and Eq.~\eqref{eq:QFI_R2}, the inequality holds:
\[
\mathcal{F}_\omega(\R(T)) \leq 2\left(\mathcal{B}\left(\R(T),\int_{0}^{T/2} dt \mathcal{T} e^{\int_t^T d\tau \mathcal{L}(\tau)} i[\R(t),\an^\dagger\an]\right)+\mathcal{B}\left(\R(T),\int_{T/2}^{T} dt \mathcal{T} e^{\int_t^T d\tau \mathcal{L}(\tau)} i[\R(t),\an^\dagger\an]\right)\right).
\]
By generalizing this inequality to an arbitrary integer $M$, we obtain the following:
\[\begin{split}
    \mathcal{F}_\omega(\R(T)) &\leq M \sum_{j=1}^M \mathcal{B}\left(\R(T),\int_{\tfrac{j-1}{M}T}^{\tfrac{j}{M}T} dt \mathcal{T} e^{\int_t^T d\tau \mathcal{L}(\tau)} i[\R(t),\an^\dagger\an]\right) \\
&= M \sum_{j=1}^M \mathcal{B}\left(\R(T), \frac{T}{M} \mathcal{T} e^{\int_t^T d\tau \mathcal{L}(\tau)} i[\R(\tfrac{jT}{M}),\an^\dagger\an]
\right)\\
&= T \frac{T}{M} \sum_{j=1}^M \mathcal{B}\left(\R(T), \mathcal{T} e^{\int_t^T d\tau \mathcal{L}(\tau)} i[\R(\tfrac{jT}{M}),\an^\dagger\an]
\right),
\end{split}
\]
whose summation becomes integration at the limit $M\to \infty$, such that,
\[
\label{eq:F_B_T_bound}
    \mathcal{F}_\omega(\R(T)) \leq T \int_0^T dt  \mathcal{B}\left(\R(T), \mathcal{T} e^{\int_t^T d\tau \mathcal{L}(\tau)} i[\R(t),\an^\dagger\an]\right).
\]
Finally, from the monotonicity property of the QFI under a quantum channel,
\[\begin{split} 
\mathcal{B}\left(\R(T), \mathcal{T} e^{\int_t^T d\tau \mathcal{L}(\tau)} i[\R(t),\an^\dagger\an]\right)
=\mathcal{B}\left(\mathcal{T}e^{\int_t^T d\tau \mathcal{L}(\tau)}\rho(t),\mathcal{T}e^{\int_t^T d\tau \mathcal{L}(\tau)} i[\hat\rho(t),\an^\dagger\an]\right) \leq \mathcal{B}\left(\hat\rho(t),i[\hat\rho(t),\an^\dagger\an]\right).
\end{split}
\]
In order to show this, we consider the following identity between $\mathcal{B}$ and QFI with respect to $\theta$:
\[\label{eq:B_eq_F}
\mathcal{B}\left(\hat\rho(t),i[\hat \rho(t),\an^\dagger\an]\right)=\mathcal{F}_\theta(\R_\theta(t)),
\]
where $\R_\theta(t)=\exp(-i\theta\an^\dagger \an)\R(t)\exp(i\theta\an^\dagger \an )$.  Now, let us make use of the relationship between the QFI and the quantum fidelity,
\[
\mathcal{F}_\theta(\R_\theta(t)) = 4\frac{d^2}{d\theta^2}(1 - F(\R_{-d\theta/2}(t),\R_{d\theta/2}(t)) ), 
\]
where the fidelity between two quantum states is defined as $F(\R,\R') = \tr(\sqrt{\sqrt{\R}\R'\sqrt{\R}})^2$.
By using the fact that the fidelity is monotonic increasing under a quantum channel, which includes the time evolution under Lindblad dynamics, we obtain
\[\begin{split}\label{eq:ineq_lindblad_fidelity}
    \mathcal{F}_\theta \left(\mathcal{T}e^{\int_t^T d\tau \mathcal{L}(\tau )}\R_\theta(t)\right) &= 4\frac{d^2}{d\theta^2}\left(1 - F\left(\mathcal{T}e^{\int_t^T d\tau \mathcal{L}(\tau)}\R_{-d\theta/2}(t),\mathcal{T}e^{\int_t^T d\tau \mathcal{L}(\tau)}\R_{d\theta/2}(t)\right) \right)\\
    &\leq 4\frac{d^2}{d\theta^2}(1 - F(\R_{-d\theta/2}(t),\R_{d\theta/2}(t))  \\
    &=\mathcal{F}_\theta(\R_\theta(t)).
\end{split}
\]
Moreover, we can find that
\[\label{eq:B_(T,T)=F}
\mathcal{F}_\theta \left(\mathcal{T}e^{\int_t^T d\tau \mathcal{L}(\tau )}\R_\theta(t)\right) = \mathcal{B}\left(\mathcal{T}e^{\int_t^T d\tau \mathcal{L}(\tau)}\hat\rho(t),\mathcal{T}e^{\int_t^T d\tau \mathcal{L}(\tau)} i[\hat\rho(t),\an^\dagger\an]\right)  =\mathcal{B}\left(\R(T), \mathcal{T} e^{\int_t^T d\tau \mathcal{L}(\tau)} i[\R(t),\an^\dagger\an]\right).
\]
Combining Eq.~\eqref{eq:B_eq_F}, Eq.~\eqref{eq:ineq_lindblad_fidelity}, and Eq.~\eqref{eq:B_(T,T)=F}, we have the inequality,
\[
\label{eq:B_bound_monotonicity}
\mathcal{B}\left(\R(T), \mathcal{T} e^{\int_t^T d\tau \mathcal{L}(\tau)} i[\R(t),\an^\dagger\an]\right) \leq \mathcal{B}\left(\hat\rho(t),i[\hat\rho(t),\an^\dagger\an]\right).
\]
By combining, Eqs.~\eqref{eq:F_B_T_bound} and \eqref{eq:B_bound_monotonicity}, we obtain
\[\begin{split}\label{eq:ineq_QFI_integral}
    \mathcal{F}_\omega(\R(T)) &\leq T \int_0^T dt  \mathcal{B}\left(\R(T), \mathcal{T} e^{\int_t^T d\tau \mathcal{L}(\tau)} i[\R(t),\an^\dagger\an]\right)\\
&=T\int_0^T dt \mathcal{B}\left(\mathcal{T}e^{\int_t^T d\tau \mathcal{L}(\tau)}\hat\rho(t),\mathcal{T}e^{\int_t^T d\tau \mathcal{L}(\tau)} i[\hat\rho(t),\an^\dagger\an]\right) \\
    &\leq T\int_0^T dt \mathcal{B}(\hat\rho(t),i[\hat\rho(t),\an^\dagger\an]) \\
    &= T\int_0^T dt  \mathcal{F}_\theta(\R_\theta(t)).
\end{split}
\]
As a result, the time-ordered integral is no longer required, and the QFI is upper bounded by the integral of the instantaneous QFI over time, $\mathcal{F}_\theta(\R_\theta(t))$ for $\R_\theta(t)=\exp(-i\theta\an^\dagger \an)\R(t)\exp(i\theta \an^\dagger \an)$. As a quantum state $\R(t)$ at every time $t$ is described by a squeezed thermal state, its QFI can be directly evaluated as a function of mean photon number $\mu$ and squeezing $r$~\cite{Pinel2013},
\[
\mathcal{F}_\theta(\R_\theta(t))= \frac{8\mu(t)}{2\mu(t) + 1}\sinh^2 2r(t)\leq 4\sinh^2 2r(t).
\]
By combining this with Eq.~\eqref{eq:ineq_QFI_integral}, we obtain the upper bound of the total QFI,
\begin{equation}
        \mathcal{F}_\omega(\R(T)) \leq 4T\int_0^T dt \sinh^2 2r(t),
    \end{equation}
which completes the proof.
\end{proof}

\end{document}